\documentclass[letterpaper,10pt]{article}
\usepackage{amsfonts,amsmath,amssymb,amsthm,bbm,mathtools,eurosym,subfig,graphicx,tikz,geometry,setspace,epsfig,epstopdf,float,caption}
\usepackage[dcucite]{harvard}
\newtheorem{lemma}[]{Lemma}
\newtheorem{theorem}[]{Theorem}
\newtheorem{corollary}[]{Corollary}
\newtheorem{assumption}[]{Assumption}
\newtheorem{definition}[]{Definition}
\renewcommand{\cite}{\citeasnoun}
\numberwithin{equation}{section}
\begin{document}
\title{Finite Sample Inference for the Maximum Score Estimand\thanks{We thank Tim Armstrong, Federico Bugni, Bryan Graham, Michal Koles{\'a}r,  Sokbae Lee, Matt Masten, Francesca Molinari, Ulrich M{\"u}ller, Whitney Newey, Pepe Montiel Olea, Chris Sims, and seminar and conference participants at Columbia, Duke, Penn State, Princeton, Vanderbilt, Yale, the 2019 Triangle Econometrics Conference, the 2019 Southern Economic Association Conference, the 2019 California Econometrics Conference, the 2019 Young Econometricians Conference, the Cemmap/WISE Workshop on  Advances in Econometrics, and the University of Tokyo Workshop on Advances in Econometrics for insightful comments and discussion. Kuong (Lucas) Do and Cheuk Fai Ng provided excellent research assistance. Financial support from the Economic and Social Research Council ESRC Large Research Grant ES/P008909/1 to the Centre for Microdata Methods and Practice is gratefully acknowledged.}}
\author{Adam M. Rosen \\ Duke University and CeMMAP\thanks{Address: Adam Rosen, Department of Economics, Duke University, 213 Social Sciences Box 90097, Durham, NC 27708; Email: adam.rosen@duke.edu}
\and Takuya Ura \\ University of California, Davis\thanks{Address: Takuya Ura, Department of Economics, University of California, Davis, One Shields Avenue, Davis, CA 95616; Email: takura@ucdavis.edu} }
\maketitle
\begin{abstract}
We provide a finite sample inference method for the structural parameters of a semiparametric binary response model under a conditional median restriction originally studied by Manski (1975, 1985). Our inference method is valid for any sample size and irrespective of whether the structural parameters are point identified or partially identified, for example due to the lack of a continuously distributed covariate with large support. Our inference approach exploits distributional properties of observable outcomes conditional on the observed sequence of exogenous variables. Moment inequalities conditional on this size $n$ sequence of exogenous covariates are constructed, and the test statistic is a monotone function of violations of sample moment inequalities. The critical value used for inference is provided by the appropriate quantile of a known function of $n$ independent Rademacher random variables.   Simulation studies compare the performance of the test to two alternative tests using an infeasible likelihood ratio statistic and Horowitz's (1992) smoothed maximum score estimator.
\begin{description}
\item JEL classification: C12 C14.
\item Keywords: Finite sample inference, Maximum score estimation, Moment inequalities, Partial identification.
\end{description}
\end{abstract}

\newpage

\section{Introduction}

In Chapter 41 of Volume 4 of the \emph{Handbook of Econometrics} on the estimation of semiparametric models, \cite{powell1994estimation} on page 2488 cites \cite{manski:1975} as the earliest example of semiparametric analysis of limited dependent variable models. \cite{manski1985semiparametric} provided further analysis for the binary outcome version of the model, in which the outcome is determined by the linear index threshold-crossing specification
$$
Y=1\{X\beta +U\geq 0\} \text{,}  
$$
for observable variables $Y\in \{ 0,1\} $ and $X$ a row vector in $\mathbb{R}^{K}$, where the unobservable variable $U$ is restricted to satisfy the zero conditional median restriction 
$$
\mathrm{median}\left(U\mid X\right)=0.
$$
This semiparametric model is thus distribution-free with regard to unobservable $U$.\footnote{As noted in \cite{manski1985semiparametric}, his analysis easily generalizes to cover the restriction that the conditional $\tau$th quantile of $U$ given $X$ is $0$, where $\tau\in(0,1)$ is known. The analysis in this paper can be similarly generalized.} Full stochastic independence between $U$ and $X$ is not required, allowing for the conditional distribution of $U$ given $X=x$ to vary with the conditioning value $x$, and thus accommodating general forms of heteroskedasticity. Under both a rank condition and a large support condition on a continuous regressor \cite{manski1985semiparametric} established point identification of $\beta $ as well as the large deviations convergence rate of the maximum score estimator.

Several further analyses of the maximum score and similar estimators for this and closely related semiparametric binary response models have since been provided, and the literature on the \emph{asymptotic} properties of the maximum score estimator is now vast. \cite{kim1990} showed that the convergence rate of the maximum score estimator is cube-root and established its nonstandard asymptotic distribution after appropriate centering and scaling. \cite{horowitz1992smoothed} developed a smoothed maximum score estimator that converges faster than the $n^{-1/3}$ rate and is asymptotically normal under some additional smoothness assumptions. Additional papers that study large sample estimation and inference applicable in the maximum score context include \cite{Manski/Thompson:1986}, \cite{Delgado/Rodriguez-Poo/Wolf:01}, \cite{Abrevaya/Huang:2005}, \cite{Leger/MacGibbon:2006}, \cite{KOMAROVA201314}, \cite{blevins2015non}, Jun, Pinkse,  and Wan (2015, 2017)\nocite{jun/pinkse/wang:2015,jun/pinkse/wang:2017}, Chen and Lee (2018, 2019)\nocite{Chen/Lee:2018}\nocite{chen/lee:2017}, \cite{Patra/Seijo/Sen:18},  \cite{seo2018}, \cite{cattaneo/jansson/nagasawa:2018}, and \cite{mukherjee2019nonstandard}.

In contrast to prior approaches for inference on $\beta$ that employ asymptotic distributional approximations, in this paper we develop a method for conducting \emph{finite sample} inference on $\beta$. To do this we employ a conditional moment inequality characterization of the observable implications of the binary response model in the finite sample. Moment inequality characterizations of the model's implications have been previously used by \cite{KOMAROVA201314}, \cite{blevins2015non}, and \cite{chen/lee:2017}, but none of these papers proposed an inference method which is valid  in the finite sample.  As was the case in the analysis provided in these papers, we do not require that $\beta$ is point identified.  For instance, we do not require that any component of $X$ is continuously distributed, much less with large support.

In fact, even if $\beta$ is point identified, and regardless of the support of $X$ in the population, the set of observed values of $X$ in any finite sample is discrete.  Indeed, \cite[page 320]{manski1985semiparametric} defines, ``the maximum score estimate $\hat{B}_n$ to be the \emph{set} of solutions to the problem $\max_{b \in \mathcal{B}}S_n(b)$'' where $\mathcal{B}$ is the parameter space and $S_n(\cdot)$ denotes the sample score function.\footnote{Manski (1985) used $B$ to denote the parameter space and upper case $N$ to denote sample size, which we have changed to $\mathcal{B}$ and lower case $n$ to match our notation.}  He shows that if the given sufficient conditions for point identification hold, then the distance between $\hat{B}_n$ and $\beta$ converges almost surely to zero, implying consistency of any sequence of $\hat{\beta}_n \in \hat{B}_n$ for $\beta$. Intuitively, the set of possible maximum score point estimators shrinks to a point as $n \rightarrow \infty$.  Given that our aim in this paper is to conduct finite sample inference, we must own up to the fact that even if $\beta$ is point identified, there is a proper set of values $b$ to which $\beta$ is observationally equivalent on the basis of only values of $X$ observed in the finite sample.

We thus introduce the concept of the \emph{finite sample identified set} as the set of parameter vectors $b \in \mathcal{B}$ that satisfy the observable implications of the binary response model \emph{conditional} on a size $n$ sequence of observable covariate vectors $\mathcal{X}_n \equiv (X_1,\ldots,X_n)$.  This set can be thought of as what the population identified set would be if the support of $X$ were in fact $\mathcal{X}_n$.  It differs from the set estimator $\hat{B}_n$ because it is not a function of the outcomes $Y_1,....,Y_n$, instead only incorporating the observable implications of the model for the distribution of outcomes conditional on $\mathcal{X}_n$.  This newly defined finite sample identified set is useful because our finite sample inference approach is driven only by observable implications regarding $Y_1,....,Y_n$ conditional on $\mathcal{X}_n$, and will be explicit in not being able to detect violations of conditional moment inequalities that condition on values of $X$ not observed in the sample.

The approach taken here exploits the implication of the binary response model that conditional on $\mathcal{X}_n$, each outcome $Y_i$ is distributed Bernoulli with parameter $p\left(X_i,\beta \right) \equiv \mathbb{P}\left(U_i \geq -X_i \beta \mid \mathcal{X}_n\right)$. In practice the Bernoulli probabilities $p\left(X_i,\beta \right)$ are unknown. Nonetheless, conditional on $\mathcal{X}_n$, each $p\left(X_i,\beta \right) $ is bounded from above or below by $1/2$ according to the sign of $X_i \beta $. Consequently, for any known nonnegative-valued function $g\left( \cdot \right) : \mathcal{X}_n \rightarrow \mathbb{R}$, the finite sample distributions of $\omega_{ui}(\beta,g) \equiv (2Y_i-1) 1\{X_i\beta \geq 0\}g\left( X_i \right) $ and $\omega_{li}(\beta,g) \equiv (1-2Y_i)1\{X_i\beta \leq 0\}g\left( X_i\right) $ conditional on $\mathcal{X}_n$ can be bounded from below.  The test statistic $T_n(b)$ that we use to implement our test of the null hypothesis $H_0: \beta=b$ is a supremum of weighted sample averages of $-\omega_{ui}(\beta,g)$ and $-\omega_{li}(\beta,g)$, where the supremum is taken over particular collections of functions $g\left( \cdot \right)$.  The test statistic $T_n(b)$ is shown to be bounded from above by a function $T^{\ast}_n\left(b\right)$ of $n$ independent Rademacher random variables, such that the finite sample distribution of $T^{\ast}_n\left(b\right)$ given $\mathcal{X}_n$ is known. Then, under the null hypothesis $\beta = b$, we have 
$$
\mathbb{P}\left(T_{n}(b)>q_{1-\alpha }\mid \mathcal{X}_n\right)\leq \alpha,
$$ 
where $q_{1-\alpha }$ is the $1-\alpha $ quantile of $T^{\ast}_n\left(b\right)$ given $\mathcal{X}_n$. We establish that if particular functions $g\left( \cdot \right)$ are used, the moment functions which $T_{n}(b)$ incorporates preserve all the identifying information contained in the finite sample identified set. We further establish a power result for alternatives that lie outside this set.

For the sake of comparison we also consider likelihood ratio tests. Optimality of the likelihood ratio test for the null hypothesis of $\beta=b$ depends on the specification of the alternative hypothesis. The likelihood ratio test is a most powerful test when the alternative hypothesis is simple, meaning that it specifies a unique conditional distribution $\mathbb{P}\left(Y_1,...,Y_n \mid \mathcal{X}_n \right)$. However, when the alternative hypothesis is of the form $\beta\ne b$ or $\beta=\tilde{b}$, many possible distributions of $\mathbb{P}\left(Y_1,...,Y_n \mid \mathcal{X}_n \right)$ are permitted under the conditional median restriction. These alternatives are therefore composite. The likelihood ratio test of the null hypothesis $\beta = b$ against either of these alternatives corresponds to a coin flip independent of the data, and has only trivial power against any distribution under either of these composite alternative hypotheses. We establish a lower bound on the power of our test, and we show that this power bound is increasing in the degree to which a hypothesized parameter vector $b$ violates the inequalities that characterize the finite sample identified set. Thus values of $b$ that are sufficiently far from the finite sample identified set by this measure are guaranteed to be rejected with probability exceeding the size of the test.  

Among the aforementioned papers from the literature on maximum score, the most closely related is that of \cite{chen/lee:2017}, who also cast the implications of Manski's (1985) model as conditional moment inequalities for the sake of delivering a new insight, albeit one that is entirely different from ours. \cite{chen/lee:2017} expand on the conditional moment inequalities used by \cite{KOMAROVA201314} and \cite{blevins2015non} to develop a novel conditional moment inequality characterization of the identified set which involves conditioning on two linear indices instead of on the entire exogenous covariate vector. They apply intersection bound inference from \cite{Chernozhukov/Lee/Rosen:09}\ to this conditional moment inequality characterization to achieve asymptotically valid inference. This cleverly exploits the model's semiparametric linear index restriction in order to side step the curse of dimensionality. Although a good deal of focus is given to Manski's (1985) binary response model, their method can also be applied to other semiparametric models.

To the best of our knowledge, this paper is the first to propose a finite sample inference method for $\beta$ in Manski's (1985) semiparametric binary response model. The test proposed is shown to control size for any sample size $n$. This paper is also the first to introduce the concept of a finite sample identified set, both within the context of the semiparametric binary response model and more broadly, explicitly defining the set of model parameters logically consistent with the modeling restrictions and only information that can be gathered from observable implications conditional on realizations of exogenous variables observed in the finite sample.  There are however a handful of precedents for employing finite sample inference with other partially identifying models. \cite{Manski:2007} considers the problem of predicting choice probabilities for the choices individuals would make if subjected to counterfactual variation in their choice sets. In the absence of the structure afforded by commonly used random utility models, he shows that counterfactual choice probabilities are partially identified, and proposes a procedure for inference using results from \cite{Clopper/Pearson:34}. \cite{chernozhukov/hansen/Jansson:2009} propose a finite sample inference method in quantile regression models in which the outcome is continuously distributed. Their approach exploits a ``conditionally pivotal property'' to bound the finite sample distribution of a GMM criterion incorporating moment equalities, but which does not require point identification for its validity.  \cite{Syrgkanis/Tamer/Ziani:18} conduct inference on partially identified parameters of interest in auction models imposing weak assumptions on bidders' information.  They propose a method to conduct finite sample inference on moments of functions of the underlying valuation distribution using concentration inequalities. \cite{Armstrong/Kolesar:2018} provide methods for optimal inference on average treatments that are finite sample valid in the special case in which regression errors are normal, and asymptotically valid more generally. Their conditions cover cases where identification may fail due to lack of overlap of the support of conditioning variables. The approach taken in this paper for finite sample inference in the context of Manski's (1985) binary response model is different from all of these.

The rest of this paper is organized as follows. Section \ref{Section: Testing Problem} formally sets out the testing problem and the moment inequality representation of the finite sample identified set. Section \ref{Section:Test Statistic and Critical Value} lays out the construction of the test statistic and corresponding critical value, and establishes the finite sample validity of the test. It also establishes a finite sample (lower) power bound for our test. Section \ref{Section: likelihood ratio test} considers the likelihood ratio test. Section \ref{Section: Monte Carlos} demonstrates the performance of our approach by reporting results from Monte Carlo simulations with comparison to inference using an infeasible likelihood ratio test and Horowitz's (1992) smoothed maximum score estimator. Section \ref{Section: Conclusion} concludes and discusses avenues for future research. All proofs are in the Appendix. Unless otherwise stated, our analysis throughout this paper should be read as conditional on observable covariate vectors $\mathcal{X}_n \equiv (X_1,\ldots,X_n)$.

\section{Model and Moment Restrictions}\label{Section: Testing Problem}

This section is divided into three subsections, the first of which formally presents the modeling restrictions imposed.  The second subsection describes the observable implications of the binary response model \emph{conditional} on a size $n$ sequence of covariate vectors, $\mathcal{X}_n$, in contrast to those observable implications obtainable from knowledge of the population distribution of observable variables.  Based on these observable implications, this second subsection introduces our definition of the finite sample identified set. It clarifies what violations of our model's implications the proposed test can feasibly detect, which is useful for power considerations.  The third subsection uses the special case in which covariates $X_i$ are bivariate as an illustrative example. The developments of Subsections \ref{Subsection: FSIDset} and \ref{Section: Two Covariates} are however not essential for establishing size control of the test presented in Section \ref{Section:Test Statistic and Critical Value}.

\subsection{Model}

The following assumption formalizes the restrictions of the semiparametric
binary response model under study and the requirements on the sampling
process. We maintain these assumptions throughout this paper. 

\begin{assumption}\label{Assumption: Model and Sampling Process}
(i) Random vectors $\{ \left(Y_i,X_i,U_i\right) :i=1,\ldots,n\} $ reside on a probability space $\left( \Omega ,\mathfrak{F} ,\mathbb{P}\right) $, where $\mathfrak{F} $ contains the Borel sets on $\Omega $. 
(ii) Variables $\{\left(Y_i,X_i\right) :i=1,\ldots,n\} $ are observed. 
(iii) There is a column vector $\beta\in\mathbb{R}^{K}$ such that $\mathbb{P}\left( Y_i=1\{ X_i\beta +U_i\geq 0\}\mid \mathcal{X}_n\right) =1$ and $\mathbb{P}\left( U_i\geq 0\mid \mathcal{X}_n\right) =1/2$ for every $i=1,\ldots ,n$, where $\mathcal{X}_n \equiv (X_1,\ldots,X_n)$. 
(iv) There is a known set $\mathcal{B}\subseteq \mathbb{R}^{K}$ to which $\beta $ belongs. 
(v) The unobservable variables $\left(1\{U_1\geq 0\},\ldots,1\{U_n\geq 0\}\right)$ are mutually independent conditional on $\mathcal{X}_n$.
\end{assumption}

The requirements of Assumption \ref{Assumption: Model and Sampling Process} are slightly weaker than the assumptions used in the existing literature (e.g., Manski, 1975, 1985). Parts (i), (ii) and (iv)\ are standard. Although it is not necessary in this paper because we employ partial identification analysis, the parameter space $\mathcal{B}$ can be restricted by imposing one of the usual scale normalizations from the literature, such as $|b_1|=1$ for all $b \in \mathcal{B}$. Part (iii) imposes the binary response structure and the requirement that $\mathbb{P}\left( U_i\geq 0\mid \mathcal{X}_n\right)=1/2$ for each $i$. Binary response models typically require that $U_i$ is continuously distributed in a neighborhood of zero, in which case this is implied by the usual conditional median restriction. Strictly speaking, we do not need to impose that each $U_i$ is continuously distributed at zero, and hence we replace the median restriction with this requirement. Part (v) holds if $\left(Y_i,X_i,U_i\right) $ are independent and identically distributed, but is much more general. The observations $\{ \left( Y_i,X_i\right):i=1,\ldots,n\} $ can be non-independent and non-identically distributed. Indeed, even if $\{ \left( Y_i,X_i\right):i=1,\ldots,n\} $ were restricted to be i.i.d., it would remain the case that $\{Y_i: i=1,\ldots,n\}$ would not be i.i.d. conditional on $\mathcal{X}_n$. Throughout the text, $\mathbb{E}[\cdot]$ is used to denote population expectations taken with respect to $\mathbb{P}$, and $\mathbb{E}_n[\cdot] \equiv n^{-1}\sum_{i=1}^{n}[\cdot]$.

The power result presented in Theorem \ref{theorem:power_bound} in Section \ref{Section:Test Statistic and Critical Value} and the results of Section \ref{Section: likelihood ratio test} on the likelihood ratio test additionally invoke the following assumption.
\begin{assumption}\label{Assumption: Uindendepednece}
\label{Assumption: independent U} The unobservable variables $\left(U_{1},...,U_{n}\right)$ are mutually independent conditional on $\mathcal{X}_n$.
\end{assumption}
\noindent This restriction is common in the prior literature on maximum score estimation, but is not required for many of the results in this paper, and in particular is not necessary to establish size control for our test.  The assumption is satisfied for example in models that restrict $(X_i, U_i)$ to be i.i.d.\footnote{For instance this is implied by Assumption 3 of \cite{manski1985semiparametric}.} Note however that Assumption \ref{Assumption: Uindendepednece} does not require $U_1,...,U_n$ to be i.i.d. given $\mathcal{X}_n$, but only mutually independent.

\subsection{Observable Implications Conditional on $\mathcal{X}_n$}\label{Subsection: FSIDset}

To conduct finite sample inference, we focus solely on the implications obtainable from a sequence of $n$ draws of $\left( Y,X\right) $ in a sample $\{ \left( Y_i,X_i\right) :i=1,\ldots,n\} $ and not on features of the population distribution of these variables that could only be obtained on the basis of infinitely many observations. Consequently our focus is not on the identified set that could be obtained from knowledge of the population distribution of $\left( Y,X\right) $ in an infinitely large population, but rather on the set obtainable solely from knowledge of a size $n$ sample of observations in accord with Assumption \ref{Assumption: Model and Sampling Process}. By definition, this is the set of parameter vectors $b \in \mathcal{B}$ such that the distribution of $Y_1,\ldots,Y_n$ conditional on $\mathcal{X}_n$ matches that of $\tilde{Y}_i \equiv 1\{X_i b + \tilde{U}_i \geq 0\}$ for a sequence of random variables $\tilde{U}_1,\ldots,\tilde{U}_n$ that satisfy the restrictions placed on the conditional distribution of $U_1,\ldots,U_n$ in Assumption \ref{Assumption: Model and Sampling Process}. We refer to this set as the finite sample identified set and denote it as $\mathcal{B}_{n}^{\ast }$. 

\begin{definition}\label{Def: Finite Sample Identified Set}
The \textbf{finite sample identified set} for $\beta$, denoted $\mathcal{B}_{n}^{\ast }$, is the set of $b\in\mathcal{B}$ for which there exist random variables $\{\tilde{Y}_{i} :i=1,\ldots,n\} $ and $\{\tilde{U}_{i} :i=1,\ldots,n\} $ such that:
\begin{itemize} 
\item[(i):]  $\mathbb{P}\left( \tilde{Y}_i=1\{ X_ib +\tilde{U}_{i}\geq 0\}\mid \mathcal{X}_n\right) =1$,
\item[(ii):] $(\tilde{Y}_{1},\ldots,\tilde{Y}_{n}) $ and $\left(Y_{1},\ldots,Y_{n}\right)$ have the same distribution conditional on $\mathcal{X}_n$,
\item[(iii):] $\mathbb{P}\left(\tilde{U}_{i}\geq 0\mid \mathcal{X}_n\right) =1/2$ for every $i=1,\ldots ,n$,
\item[(iv):] $\{1\{\tilde{U}_{i}\geq 0\}:i=1,\ldots ,n\}$ are mutually independent given $\mathcal{X}_n$.
\end{itemize}
\end{definition}

Our next task is to express $\mathcal{B}_{n}^{\ast }$ with a moment inequality representation useful for inference. The following lemma sets out two observable implications that will be useful for this purpose. 

\begin{lemma}
\label{Lemma:Conditional Moment Inequalities} 
Under Assumption \ref{Assumption: Model and Sampling Process}, 
\begin{eqnarray} \label{positive if inequality}
X_i\beta  &\geq &0\implies \mathbb{E}[2Y_i-1\mid \mathcal{X}_{n}]\geq 0%
\text{,}   \\ \label{negative if inequality}
X_i\beta  &\leq &0\implies \mathbb{E}[2Y_i-1\mid \mathcal{X}_{n}]\leq 0%
\text{.}
\end{eqnarray}
\end{lemma}

From the inequalities of the lemma, it further follows that if $X_i \beta = 0$ then $\mathbb{E}[2Y_i-1\mid \mathcal{X}_{n}]=0$.
Moreover (2.1) and \eqref{negative if inequality} and the implications of them described above hold with $\beta$ replaced by any $b$ that is an element of the finite sample identified set $\mathcal{B}_{n}^{\ast }$. This can be proven by following precisely the same steps as in the proof of the lemma with $\tilde{U}_i$ from Definition 1 replacing $U_i$.

With Lemma \ref{Lemma:Conditional Moment Inequalities} in hand, the following theorem provides a moment inequality characterization of the finite sample identified set.
\begin{theorem}
\label{Theorem: Finite Sample Identified Set} 
Under Assumption \ref{Assumption: Model and Sampling Process}, the finite sample identified set for $\beta $ is 
$$
\mathcal{B}_{n}^{\ast }=\left\{ b\in \mathcal{B}: 
\mathbb{E}\left[ \left( 2Y_i-1\right) 1\{X_ib\geq 0\} \mid \mathcal{X}_n\right] \geq 0\geq \mathbb{E}\left[ \left( 2Y_i - 1\right) 1\{X_ib\leq 0\} \mid \mathcal{X}_n\right]\mbox{ for every }i=1,\ldots,n
\right\} .  
$$
\end{theorem}

The conditional moment inequalities characterizing $\mathcal{B}_{n}^{\ast }$ in Theorem \ref{Theorem: Finite Sample Identified Set} are equivalent to (\ref{positive if inequality}) and (\ref{negative if inequality}) for all $i=1,\ldots,n$. However, using this conditional moment inequality representation to conduct inference on $\beta$ is complicated by the fact that in a sample of $n$ observations the distribution of $Y_i$ given $\mathcal{X}_n$ can vary across $i$, even if $(Y_i,X_i): i=1,\ldots,n$ are identically distributed, and there is only one observation of $(Y_i,X_i)$ for each $i$.

Some level of aggregation of these implications across $i$ is therefore required. One way to do this is to interact the expressions inside the conditional expectation operators in Theorem \ref{Theorem: Finite Sample Identified Set} with nonnegative-valued functions of the exogenous covariates and take sample averages. For this purpose, define $\{g_u(\cdot, v):v\in \mathcal{V}_u\}$ and $\{g_{l}(\cdot, v):v\in \mathcal{V}_l\}$ to be collections of such non-negative instrument functions indexed by $v\in \mathcal{V}_u$ and $v \in \mathcal{V}_l$, respectively.  That is, for any $v \in \mathcal{V}_u$, $g_u(\cdot, v):\mathbb{R}^K \rightarrow \mathbb{R}_+$ is an instrument function mapping from $K$ dimensional Euclidean space (on which each $X_i$ resides) to the nonnegative reals. Likewise, for each $v \in \mathcal{V}_l$, the function $g_{l}(\cdot, v)$ also maps to the nonnegative reals. Shortly, particular collections of such functions and corresponding index sets $\mathcal{V}_u$ and $\mathcal{V}_l$ will be defined for construction of our proposed test statistic.

Since the instrument functions are nonnegative-valued, interacting them with the expressions inside the conditional expectations in Theorem \ref{Theorem: Finite Sample Identified Set} and averaging across $i$ will preserve the sign of the conditional expectation. Specifically, consider for any nonnegative valued instrument function $g\left(\cdot \right)$ the conditional moment inequalities 
\begin{eqnarray}
\mathbb{E}\left[\mathbb{E}_n\left[ (2Y-1) 1\{Xb\geq 0\}g(X)\right]\mid \mathcal{X}_n \right] \geq 0\text{,}  \label{Conditioning Set Expectation pos} \\
\mathbb{E}\left[\mathbb{E}_n\left[ (1-2Y) 1\{Xb\leq 0\}g(X)\right]\mid \mathcal{X}_n\right] \geq 0\text{.}\label{Conditioning Set Expectation neg}
\end{eqnarray}
These inequalities are valid for all $b \in \mathcal{B}_{n}^{\ast }$ because they are implications of the conditional moment inequalities characterizing $\mathcal{B}_{n}^{\ast }$ in Theorem \ref{Theorem: Finite Sample Identified Set}. Indeed, they are both valid for any \emph{collection} of nonnegative-valued instrument functions. A potential drawback to aggregation of the conditional moments however is that (\ref{Conditioning Set Expectation pos}) and (\ref{Conditioning Set Expectation neg})\ for just any such collection of positive instrument functions need not in general fully characterize $\mathcal{B}_{n}^{\ast }$, so that using the latter inequalities can result in a loss of identification power.

Particular collections $\{g_u(\cdot,v):v\in \mathcal{V}_u\} $ and $\{g_{l}(\cdot,v) :v\in \mathcal{V}_l\} $ are now defined so that imposing \eqref{Conditioning Set Expectation pos} and \eqref{Conditioning Set Expectation neg} ensures preservation of the full identifying power of the conditional-on-$\mathcal{X}_n$ moment inequalities in Theorem \ref{Theorem: Finite Sample Identified Set}.  These collections differ from those used by \cite{andrews/shi:2013} for translating the identifying power of conditional moment inequalities to a collection of unconditional moment inequalities.  In the present setting, there is no issue of converting inequalities conditional on continuous variables to unconditional ones, because the conditioning set in the inequalities characterizing $\mathcal{B}^{\ast}_n$ is finite. Instead, the problem to be addressed is how best to aggregate these implications across observations $i$ given the non-i.i.d. nature of $Y_i$ conditional on $\mathcal{X}_n$.  In constructing our collection of information-preserving instrument functions, we exploit two features specific to the task at hand, namely first that our focus is on finite sample inference conditional on $\mathcal{X}_n$ and second that whether or not $\mathbb{E}\left[2Y_i-1 \mid \mathcal{X}_n\right]\geq 0$ ($\leq 0$) depends only on whether the linear index $X_i\beta $ is at least (at most) zero. In the existing literature, Pinkse (1993, Section 3.3)\nocite{pinkse:1993} has used the second feature for exact computation of the maximum score estimator. 

Toward this end, consider the following two sequences of binary indicators:
\begin{eqnarray*}
r_u\left( b\right) &\equiv &\left( 1\{X_{1}b\geq 0\},\ldots,1\{X_{n}b\geq 0\} \right) \text{,} \\ 
r_{l}\left( b\right) &\equiv &\left( 1\{X_{1}b\leq 0\},\ldots,1\{X_{n}b\leq 0\}\right).
\end{eqnarray*}%
Irrespective of how $g(X)$ is defined, $\mathbb{E}_n\left[ (2Y-1) 1\{Xb\geq 0\}g(X)\right]=\mathbb{E}_n\left[ (2Y-1) 1\{Xb^{\prime}\geq 0\}g(X)\right]$ whenever $r_u(b)=r_u(b^{\prime })$\ and $\mathbb{E}_n\left[ (1-2Y) 1\{Xb\leq 0\}g(X)\right]=\mathbb{E}_n\left[ (1-2Y) 1\{Xb^{\prime}\leq 0\}g(X)\right]$ whenever $r_{l}(b)=r_{l}(b^{\prime })$. Using the functions $r_u(\cdot)$ and $r_{l}(\cdot)$, we denote by $\mathsf{V}_u$ the coimage of the function $r_u(\cdot)$ on $\mathcal{B}$, and by $\mathsf{V}_{l}$ that of the function $r_{l}(\cdot)$.\footnote{The coimage of a function $f$ is defined as the quotient set of the equivalence relation defined by $f$.} That is, taking for instance $\mathsf{V}_u$, the sequences of inequalities defining unique values for $r_u(\cdot)$, namely $X_ib\geq 0$ and $X_ib < 0$ for $i=1,\ldots,n$, partition $\mathbb{R}^K$ into the collections of sets $\mathsf{V}_u$. Similarly, the partition $\mathsf{V}_{l}$ comprises regions on which the function $r_{l}(\cdot)$ takes the same sequence of ones and zeros according to the satisfaction of inequalities $X_ib \leq 0$ and $X_ib > 0$, $i=1,\ldots,n$. The partitions $\mathsf{V}_u$ and $\mathsf{V}_{l}$ coincide with each other, with the exception of how points at which $X_ib = 0$ are assigned within each partition. This ensures that when $X_ib = 0$, the joint implication of the inequalities \eqref{positive if inequality} and \eqref{negative if inequality} are captured by our testing procedure.

Thus the partitions $\mathsf{V}_u$ and $\mathsf{V}_{l}$ are collections of sets defined by whether or not they satisfy sequences of linear inequalities $X_ib\geq 0$ and $X_ib\leq 0$, $i=1,\ldots,n$, respectively. Equivalently, each such set comprises an intersection of $n$ hyperplanes in $\mathbb{R}^K$.  Such partitions are referred to as hyperplane arrangements in the computational geometry literature.  Algorithms for their enumeration and computation have been developed, and indeed put to good use recently in econometrics by \cite{Gu/Koenker:2018} for the purpose of computing nonparametric maximum likelihood estimators for binary response models.  The model studied by \cite{Gu/Koenker:2018} differs from the one in this paper, as it allows for random coefficients but requires independence of covariates and unobservable variables, and its focus is on computation rather than finite sample inference. Nonetheless, the same developments from the computational geometry literature on the enumeration of hyperplane arrangements can be employed for efficient computation here.

Since $n$ is finite, there are only finitely many elements of each of $\mathsf{V}_u$ and $\mathsf{V}_{l}$. With regard to the inequalities (\ref{Conditioning Set Expectation pos}) and (\ref{Conditioning Set Expectation neg}), all members $v$ of any set in $\mathsf{V}_u$ and all members $v$ of any set in $\mathsf{V}_{l}$ produce the same values of $\mathbb{E}\left[\mathbb{E}_n\left[ (2Y-1) 1\{Xb\geq 0\}g_u(X,v)\right]\mid \mathcal{X}_n \right]$ and $\mathbb{E}\left[\mathbb{E}_n\left[ (1-2Y) 1\{Xb\leq 0\}g_l(X,v)\right]\mid \mathcal{X}_n\right]$ for any $b\in \mathcal{B}$. Thus it will suffice for us to work with a single representative from each set, i.e., the full identifying power is preserved as long as $\mathcal{V}_u$ and $\mathcal{V}_l$ have a representative from each element of $\mathsf{V}_u$ and $\mathsf{V}_{l}$, respectively. This is formalized in Theorem \ref{Theorem: ID set characterization} below.

In order to assess how large the partitions $\mathsf{V}_u$ and $\mathsf{V}_{l}$ can be, it is helpful to note that in our context each separating hyperplane is of the form $\{v \in \mathbb{R}^K: X_i v = 0 \}$, inducing a \emph{homogeneously linearly separable dichotomy}, as described by \cite{Cover:65}.   From Theorem 1 of that paper, the number of such dichotomies that can be constructed from $n$ points in $\mathbb{R}^K$ is bounded from above by
\begin{equation*}
C\left(n,K\right) \equiv 2 \sum\limits_{j=0}^{K-1} \binom{n-1}{j}\text{.}
\end{equation*}
This is equivalently the upper bound on the number of elements in each of the partitions $\mathsf{V}_u$ and $\mathsf{V}_{l}$, and is attained when every subset of $K$ points from the $n$ points $X_1,\ldots,X_n$ are linearly independent.\footnote{In this case $X_1,...,X_n$ are said to be in general position.}

This upper bound on the cardinality of $\mathsf{V}_u$ and $\mathsf{V}_{l}$, and thus on the number of representative points in the sets $\mathcal{V}_u$ and $\mathcal{V}_l$, can be quite large.  Fortunately their computation can be carried out by making use of the aforementioned results from computational geometry on computing an exhaustive set of hyperplane arrangements, also referred to as the vertex enumeration problem.  Examples of available algorithms include \cite{avis/fukuda:1996}, \cite{sleumer:1998}, and \cite{Rada/Cerny:2018}, as well as a novel computational method proposed by \cite{Gu/Koenker:2018} building on \cite{Rada/Cerny:2018}.\footnote{To investigate the performance of our inference method, Monte Carlo experiments are presented in Section \ref{Section: Monte Carlos}. Like most of the prior literature, we use designs with two covariates. This constitutes the simplest setting in which to investigate the performance of our inference method, and affords computational tractability. In such cases the use of sophisticated vertex enumeration algorithms is unnecessary, for reasons explained in Section \ref{Section: Two Covariates}. We thus do not investigate the use of these algorithms here, but we note their availability for higher dimensional settings.}

With partitions $\mathsf{V}_u$ and $\mathsf{V}_{l}$ now defined, the following theorem establishes a representation of the finite sample identified set $\mathcal{B}_{n}^{\ast }$ given in Theorem \ref{Theorem: Finite Sample Identified Set} that takes the form of a finite collection of inequalities of the form (\ref{Conditioning Set Expectation pos}) and (\ref{Conditioning Set Expectation neg}).

\begin{theorem}\label{Theorem: ID set characterization} 
Let Assumption \ref{Assumption: Model and Sampling Process} hold and let $b\in \mathcal{B}$. 
If $b\in\mathcal{B}_{n}^{\ast }$ then 
\begin{equation}
\forall v\in \mathcal{V}_u:\mathbb{E}\left[\ \mathbb{E}_n\left[(2Y-1) 1\{Xb\geq 0\}1\{Xv<0\}\right]\mid \mathcal{X}_n\right] \geq 0\text{,}  \label{UCon Inequality Theorem Pos}
\end{equation}
and
\begin{equation}
\forall v\in \mathcal{V}_l:\mathbb{E}\left[\ \mathbb{E}_n\left[(1-2Y) 1\{Xb\leq 0\}1\{Xv>0\}\right]\mid \mathcal{X}_n\right] \geq 0.\label{UCon Inequality Theorem Neg}
\end{equation}
Moreover, if $\mathcal{V}_u$ and $\mathcal{V}_l$ have at least one element from each member of $\mathsf{V}_u$  and $\mathsf{V}_{l}$, respectively, then \eqref{UCon Inequality Theorem Pos} and \eqref{UCon Inequality Theorem Neg} imply that $b\in \mathcal{B}_{n}^{\ast }$, so that 
$$
\mathcal{B}_{n}^{\ast } = \{b \in \mathcal{B}: \eqref{UCon Inequality Theorem Pos} \text{ and }\eqref{UCon Inequality Theorem Neg} \text{ hold}\}.
$$
\end{theorem}

The moment inequalities \eqref{UCon Inequality Theorem Pos} and \eqref{UCon Inequality Theorem Neg} are conditional on $\mathcal{X}_n$ and are thus different from those employed previously in the literature.  Our representation is perhaps most closely related to that of \cite{chen/lee:2017} for the identified set in the underlying population.  Their characterization uses inequalities that condition on the values of two linear indices in $X$: $X b$ and $X \gamma$, leading to significant dimension reduction when estimating conditional moments employed for asymptotic inference. In this paper our goal is finite sample inference, made operational by conditioning on $\mathcal{X}_n$.  Our construction leading to \eqref{UCon Inequality Theorem Pos} and \eqref{UCon Inequality Theorem Neg} exploits the finite nature of $\mathcal{X}_n$.  This is done by establishing that given $\mathcal{X}_n$, one can partition the parameter space $\mathcal{B}$ into equivalence classes $\mathsf{V}_u$ and $\mathsf{V}_{l}$ whose members comprise elements that all produce the same values of the moment functions appearing in \eqref{UCon Inequality Theorem Pos} and \eqref{UCon Inequality Theorem Neg}, respectively.  Then, using our proposed instrument functions the moment inequalities \eqref{UCon Inequality Theorem Pos} and \eqref{UCon Inequality Theorem Neg} aggregate values of $2Y_i - 1$ and $1 -2Y_i$ according to whether pairs of indices $Xb$ and $Xv$ in each of the two inequalities disagree in particular directions.

In Section \ref{Section:Test Statistic and Critical Value} we provide a test statistic that combines the extent to which these moment inequalities are violated when evaluated at any conjectured parameter vector $b \in \mathbb{R}^K$.  A confidence set for $\beta$ can then be constructed by way of test inversion.  Note however that for test inversion parameter vectors $b$ and $\tilde{b}$ need only be considered if they do not lie in the same elements of both partitions $\mathsf{V}_u$ and $\mathsf{V}_{l}$.  Any two such vectors that both reside in the same partitions will produce identical values for all moment inequalities \eqref{UCon Inequality Theorem Pos} and \eqref{UCon Inequality Theorem Neg}, and consequently the same value of the test statistic proposed in Section \ref{Section:Test Statistic and Critical Value}.

\subsection{A Simple Example: Computation with Two Covariates}\label{Section: Two Covariates}

When the covariates $X_i$ have only two components the characterization of the relevant hyperplane arrangements is greatly simplified. This makes it an ideal case in which to showcase and demonstrate the inference approach.  Because the distribution of the unobservables is nonparametrically specified, a scale normalization may be imposed, e.g., restricting the first component of $\beta$ to have absolute value of one. Due to the scale normalization, the circumstance in which there are only two components of $X_i$ is the simplest non-trivial case in which to study the semiparametric binary response model of this paper. For this reason, the case in which $\mathrm{dim}\left(X_i\right) = 2$ has appeared prominently in simulation studies of maximum score type estimators, and this is also the setting on which we focus in our Monte Carlo investigations in Section \ref{Section: Monte Carlos}.

The panels of Figure \ref{Figure: Partition Schematic} illustrate the construction of partitions $\mathsf{V}_u$ and $\mathsf{V}_{l}$ for bivariate $X_i$ in a simple example in which $n=3$. Panel (a) depicts $X_1$, $X_2$, and $X_3$, into which panel (b) additionally incorporates the values of $v_i$ for which $X_i v_i = 0$. Panel (c) drops the covariate vectors $X_i$ and panel (d) uses different colors to depict the interiors of the  elements of $\mathsf{V}_u$ and $\mathsf{V}_{l}$. Note that the partitions $\mathsf{V}_u$ and $\mathsf{V}_{l}$ only differ on their boundaries. Theorem \ref{Theorem: ID set characterization} indicates that in order to preserve the full identifying power of the finite sample identified set $\mathcal{B}^{\ast}_n$, it suffices to employ moment inequalities of the form \eqref{UCon Inequality Theorem Pos} and \eqref{UCon Inequality Theorem Neg} with sets $\mathcal{V}_u$ and $\mathcal{V}_l$ each containing one element from each of these six different colored regions.

\begin{figure}[!ht]
\centering
\caption{A schematic illustration of partitions $\mathsf{V}_u$ and $\mathsf{V}_{l}$ for $n=3$.}
\label{Figure: Partition Schematic}
\subfloat[Vectors $X_1, X_2$, and $X_3$ in $\mathbb{R}^K$.]{\includegraphics[scale = 0.4]{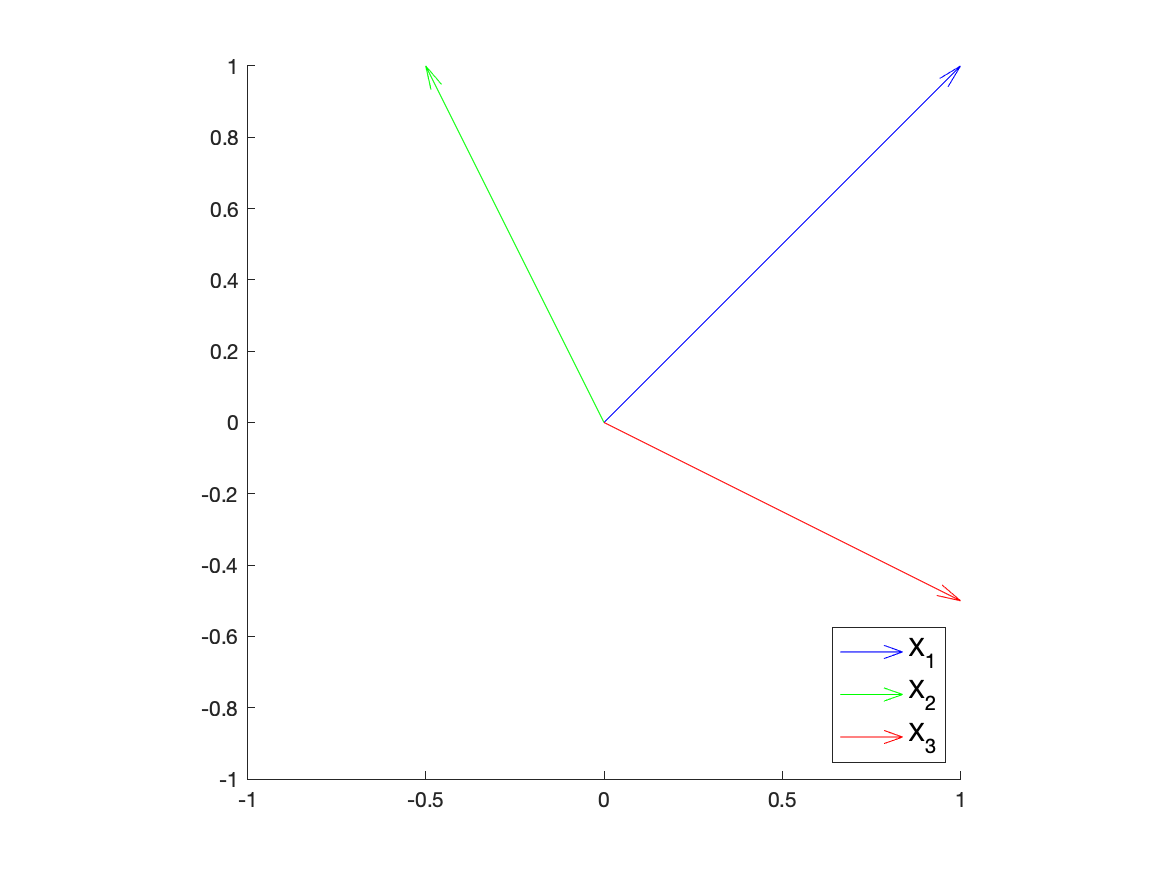}} \quad
\subfloat[$X_1, X_2, X_3$ with $v_1, v_2, v_3$ such that each $X_i v_i = 0$.]{\includegraphics[scale = 0.4]{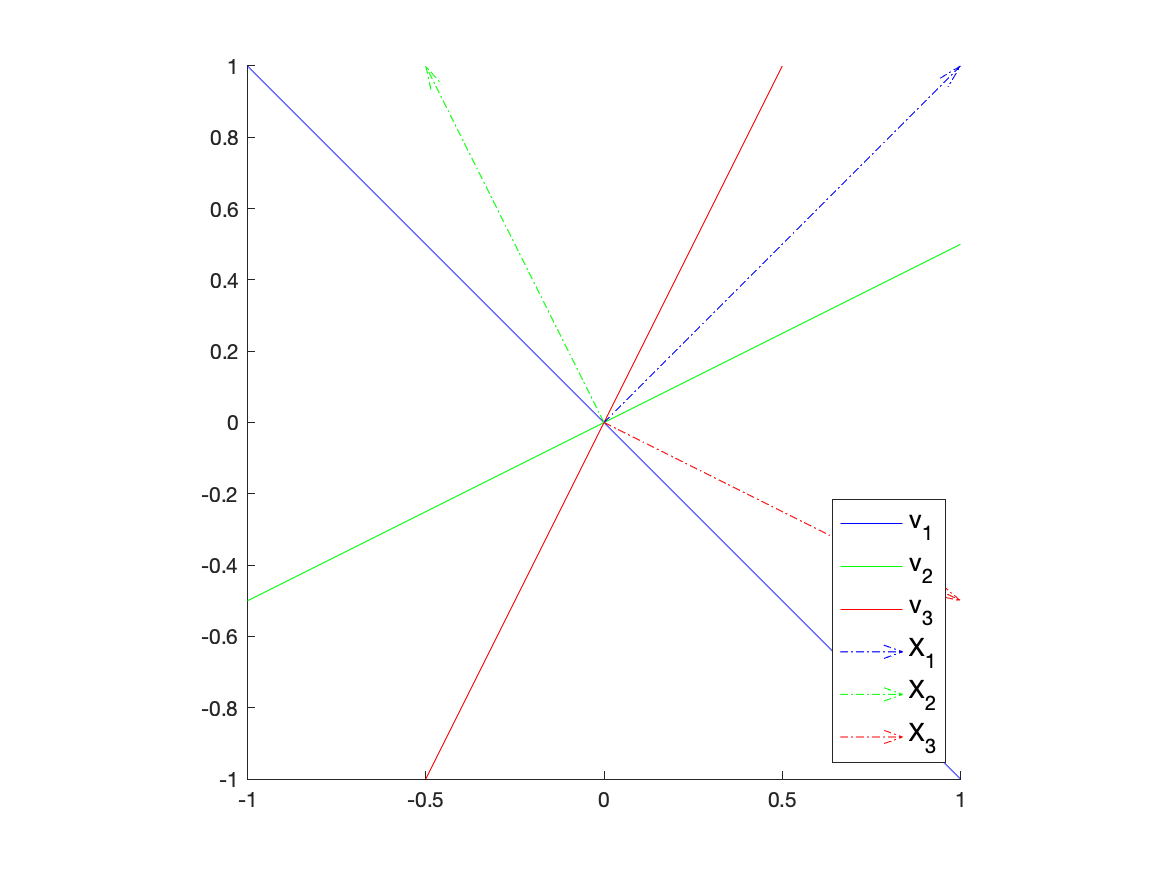}} \\
\subfloat[Lines defined by those $v_1, v_2, v_3$ for which $X_i v_i = 0$.]{\includegraphics[scale = 0.4]{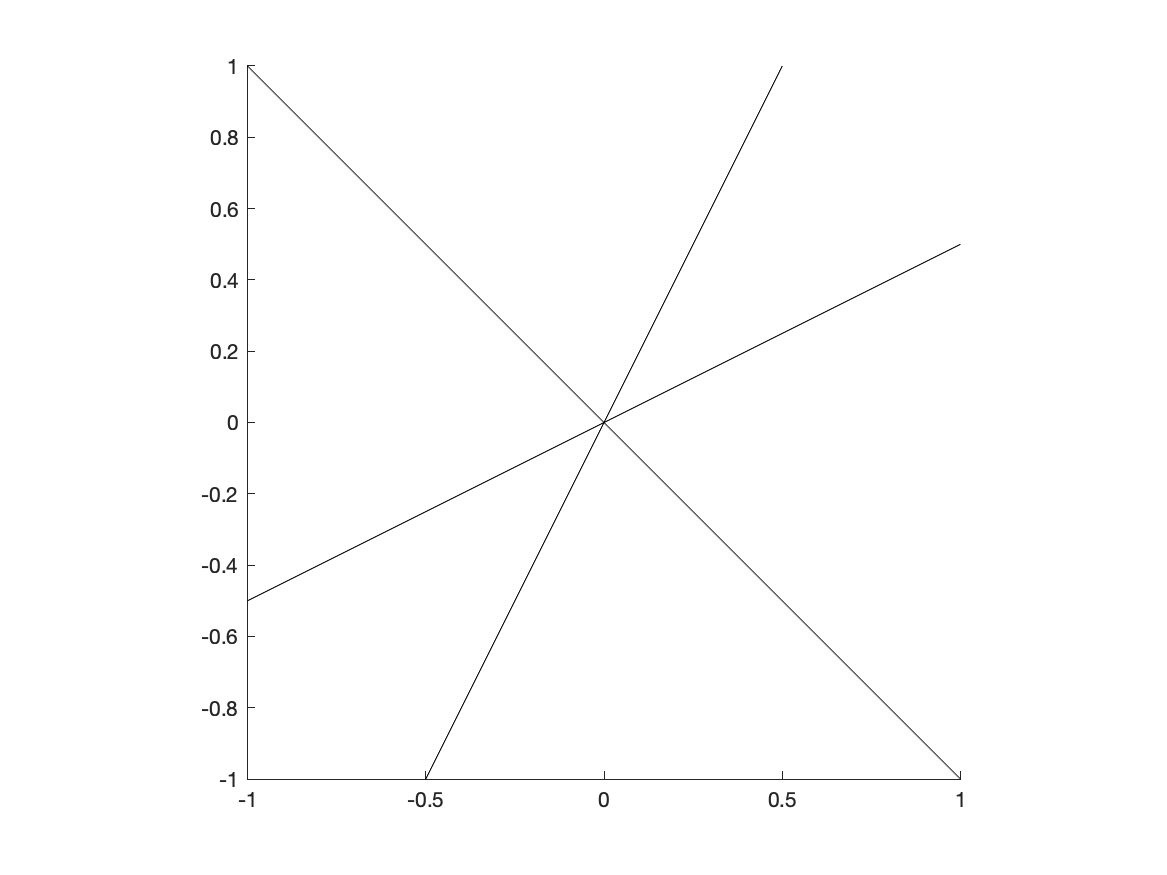}} \quad
\subfloat[Partitions of $\mathbb{R}^{K}$ defined by inequalities $X_i v_i \lessgtr 0$.]{\includegraphics[scale = 0.4]{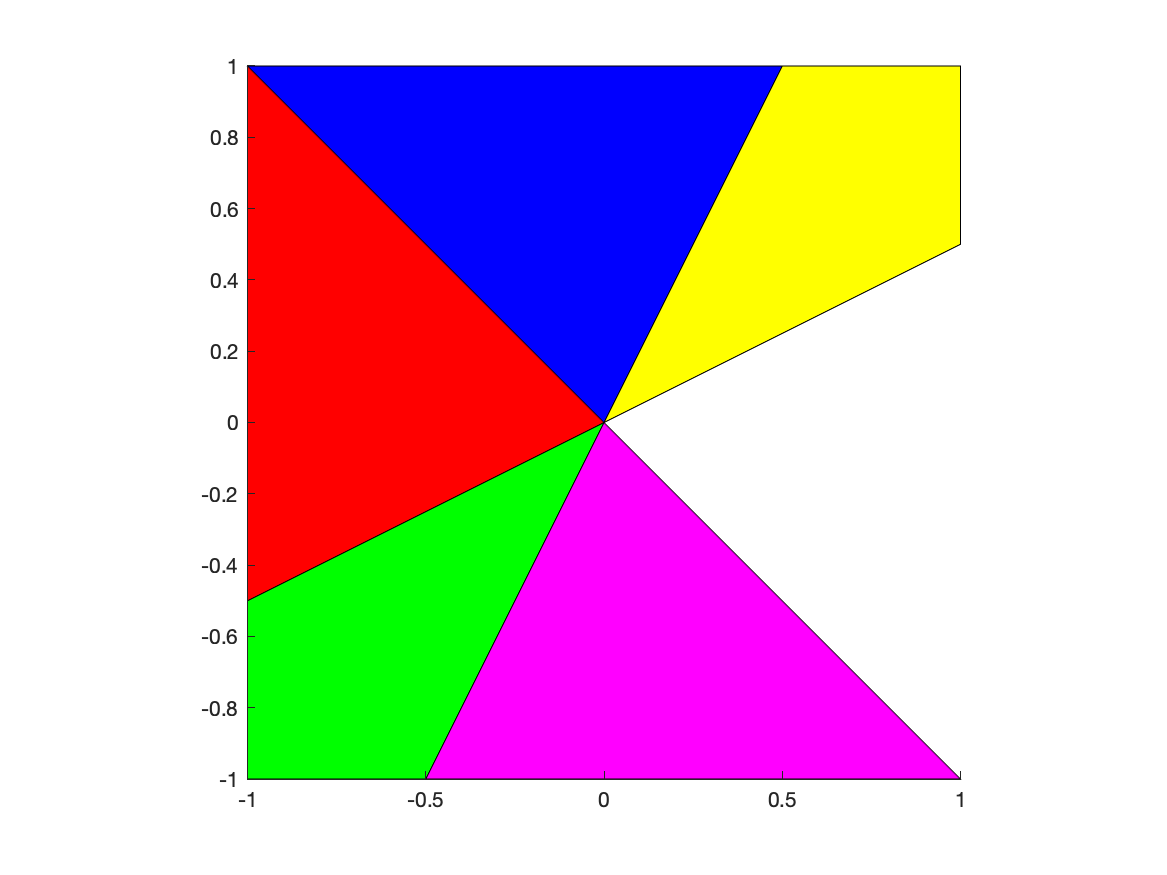}}
\end{figure}

To understand the further simplification afforded by the scale normalization, first note that for any $v$ with $v_1 \neq 0$ we can normalize the first component of $v$ such that $|v_1| = 1$ in the instrument functions that appear in \eqref{UCon Inequality Theorem Pos} and \eqref{UCon Inequality Theorem Neg} for the same reason that one can normalize the first component of the parameter vector $\beta$.\footnote{The alternative normalization that $ \left\lVert v \right\rVert = 1$ could also be used.} Specifically for any $v \in \mathbb{R}^K$ with $v_1 \neq 0$,
\begin{equation}
X_i v \gtreqqless 0 \iff X_{i1} \frac{v_1}{|v_1|} + X_{i2} \frac{v_2}{|v_1|} \gtreqqless 0\text{.}
\end{equation}
Consequently it suffices in the inequalities of \eqref{UCon Inequality Theorem Pos} and \eqref{UCon Inequality Theorem Neg} to use only values of $v$ with $|v_1| = 1$, and we impose the common scale normalization on the parameter space that $\mathcal{B}$ comprises a compact subset of $\mathbb{R}^2$ such that for all $b \in \mathcal{B}$, $|b_1|=1$. Geometrically, this means that we can select the required six representatives $v$ to form $\mathcal{V}_u$ and $\mathcal{V}_l$ from the regions illustrated in Figure \ref{Figure: Partition Schematic} panel (d) by focusing solely on values in which the first component is either $-1$ or $1$.

For the sake of actually computing such sets of representative values of $v$, possibly with much larger $n$, we may proceed making use of the normalization $|v_1| = 1$ imposed as follows. Values of $v$ that are on the boundary of satisfying the inequalities constituent in $r_u\left(v\right)$ and $r_l\left(v\right)$, namely those satisfying $X_i v = 0$, can be separated into two cases:
$$
\text{if }v_1 = +1: \quad X_i v = 0 \iff - \frac{X_{i1}}{X_{i2}} = v_2
$$
$$
\text{if }v_1 = -1: \quad Xv = 0 \iff \frac{X_{i1}}{X_{i2}} = v_2. 
$$
Thus the sequences of indicators, $r_u\left( v\right)$ and $r_{l}\left( v\right)$, depend only on where $v_2$ lies relative to the ordered sequence of values of $-Z_i$ if $v_1 = +1$ and those of $Z_i$ if $v_1 = -1$, where $Z_i \equiv \frac{X_{i1}}{X_{i2}}$.\footnote{Here it is to be understood that when $X_{i2} = 0$, $Z_i$ is defined to be $\pm \infty$ according to the sign of $X_{i1}$ if $X_{i1} \neq 0$ and $Z_i = 0$ if $X_{i1} = 0$.}

Consequently, a set of instrument values $\mathcal{V}_u$ that contains one element from each member of $\mathsf{V}_u$ can be obtained by dividing the real line into intervals according to ordered sequences of values of $-Z_i$ and $Z_i$, $i=i,...,n$, and then collecting all pairs $v = (1,v_2)$ such that $v_2$ lies in the interior of the first sequence of intervals, and all pairs $v = (-1,v_2)$ such that $v_2$ lies in the interior of the second sequence of intervals.  Specifically, let $\vartheta_1 \leq \dotsb \leq \vartheta_n$ denote the order statistics of $Z_i$, so that $\vartheta_1\equiv \min_i Z_i$ and $\vartheta_n\equiv \max_i Z_i$ and consider the following ordered sequences of intervals:
\begin{equation} \label{interval sequence upper}
\mathsf{I}_u \equiv \{(-\infty, -\vartheta_n), (-\vartheta_n, -\vartheta_{n-1}), \dotsb , (-\vartheta_2,-\vartheta_1), (-\vartheta_1,\infty ) \} \text{,}
\end{equation}
and
\begin{equation} \label{interval sequence lower}
\mathsf{I}_l \equiv \{ (-\infty, \vartheta_1), (\vartheta_1, \vartheta_2), \dotsb , (\vartheta_{n-1},\vartheta_n), (\vartheta_n,\infty ) \} \text{,}
\end{equation}
such that $\mathsf{I}_u$ and $\mathsf{I}_l$ each comprise $n+1$ open non-overlapping intervals on $\mathbb{R}$.

Recall that the inequalities in $r_u\left( v\right)$ and $r_{l}\left( v\right)$ depend only on where $v_2$ lies relative to the endpoints of these intervals. Consider a set of values of $v$, say $\mathcal{V}$, comprising a pair $(1,v_2)$ for each $v_2 \in \mathcal{I}_u$ and a pair $(-1,v_2)$ for each $v_2 \in \mathcal{I}_l$.\footnote{More precisely, the condition for $\mathcal{V}$ is (i) for every interval $I_u\in\mathcal{I}_u$, the intersection $\mathcal{V}\cap (\{1\}\times I_u)$ is nonempty, and (ii)  for every interval $I_l\in\mathcal{I}_l$, the intersection $\mathcal{V}\cap (\{-1\}\times I_l)$ is nonempty.} Then this set $\mathcal{V}$ will feature one representative from each member of the partition $\mathsf{V}_u$. Moreover, because the elements of the partition $\mathcal{V}_l$ differ from those of $\mathcal{V}_u$ only up their boundaries, such a set $\mathcal{V}$ also has one representative from each member of $\mathsf{V}_l$. Thus, using any such $\mathcal{V}$ as the sets $\mathcal{V}_u$ and $\mathcal{V}_l$ specified in \eqref{UCon Inequality Theorem Pos} will yield a moment inequality characterization of the finite sample identified set $\mathcal{B}^{\ast}_n$, i.e. with no loss of identifying power.

Such a construction for $\mathcal{V}$ is quite useful in cases where $X_i$ has two components.  A set $\mathcal{V}$ can be constructed simply by ordering the observed values of $Z_i$ ($=X_{i1}/X_{i2}$), selecting one element from each of the intervals in \eqref{interval sequence upper} and \eqref{interval sequence lower}, and pairing them with $\pm 1$ accordingly.  Computation of the test statistic proposed in the following section will require taking the supremum of a function of $v$ over the set $\mathcal{V}$, which turns out to be easy by computing $\mathcal{V}$ and taking the maximum by brute force in the case of two dimensional $X_i$. When $X_i$ has more elements, the resulting hyperplane arrangements are not as straightforward to characterize.  Nonetheless, it appears that advances from the computational geometry literature can be used to achieve computational tractability in such cases.

\section{Inference Based on Moment Inequalities} \label{Section:Test Statistic and Critical Value}
For a given value $b\in \mathcal{B}$, we consider the hypothesis test 
\begin{equation}\label{hypothesis test - two sided}
H_{0}: \beta =b \quad \text{versus} \quad H_1: \beta\neq b\text{,}
\end{equation}
on the basis of $n$ observations $\{ \left( Y_i,X_i\right) :i=1,\ldots,n\} $ following the restrictions of the semiparametric binary response model given by Assumption \ref{Assumption: Model and Sampling Process}. If one wishes to construct a confidence set for $\beta$, the set of $b$ for which $H_0$ is not rejected by a size $\alpha$ test will provide a confidence set guaranteed to contain $\beta$ with probability at least $1-\alpha$. As noted in the introduction, our method does not require point identification of $\beta $, and thus we do not assume sufficient conditions for point identification. Most notably, the existence of a continuous covariate -- much less one with full support on $\mathbb{R}$ -- is not required.

To perform inference based on moment inequalities in Theorem \ref{Theorem: ID set characterization}, we incorporate sample analogs of the moments appearing in (\ref{UCon Inequality Theorem Pos}) and (\ref{UCon Inequality Theorem Neg}), which are 
\begin{eqnarray*} 
\hat{m}_u(b,v)&\equiv&\mathbb{E}_n\left[ (2Y-1)1\{Xb\geq 0, Xv<0\} \right] \text{,\quad }v\in \mathcal{V}_{u}\text{,} \\
\hat{m}_l(b,v)&\equiv&\mathbb{E}_n\left[ (1-2Y)1\{Xb\leq 0, Xv>0\} \right]\text{,\quad }v\in \mathcal{V}_{l}\text{,}
\end{eqnarray*}
\noindent into our test statistic
\begin{equation}\label{Test statistic}
T_n\left(b\right) \equiv \max \{
0,\widehat{T}_{\mathit{u}}\left( b,\mathcal{V}_\mathit{u}\right),\widehat{T}_{\mathit{l}}\left( b,\mathcal{V}_\mathit{l}\right) \}\text{,}
\end{equation} where
\begin{eqnarray} \label{stat defs line 1}
\widehat{\sigma}_\mathit{u}^{2}(b,v) &\equiv& \mathbb{E}_n\left[ 1\{Xb \geq 0\ > Xv\} \right] - \hat{m}_{\mathit{u}}(b,v)^2\text{,} \\ \label{stat defs line 2}
\widehat{\sigma}_\mathit{l}^{2}(b,v) &\equiv& \mathbb{E}_n\left[ 1\{Xb \leq 0\ < Xv\} \right] - \hat{m}_{\mathit{l}}(b,v)^2\text{,} \\ \label{stat defs line 3}
\widehat{T}_c(b,\mathcal{V}) &\equiv& \sup_{v \in \mathcal{V}} \sqrt{n} \frac{-\hat{m}_{c}\left( b,v\right)}{\max\{\widehat{\sigma}_c(b,v), \epsilon \}}, \quad c \in \{u,l\}\text{,}
\end{eqnarray}
and $\epsilon$ is an arbitrarily small positive number taken to ensure a non-zero denominator.\footnote{In cases where either $\widehat{\sigma}_\mathit{u}^{2}(b,v)$ or $\widehat{\sigma}_\mathit{u}^{2}(b,v)$ is zero, the ratio in the definition of $\widehat{T}_c(b,\mathcal{V})$ can simply be set to $\pm \infty$ according to the sign of the numerator, or zero if the numerator is zero. The ratio that results from the use of arbitrarily small $\epsilon$ here simply serves as a placeholder for $\pm \infty$ in computations without any substantive effect.}

Instead of deriving the finite sample distribution of $T_{n}\left( b\right)$ under $H_0$, which is unknown, we construct a random variable $T^{\ast}_n\left(b\right)$ which has a known finite sample distribution given $\mathcal{X}_n$ and which satisfies
\begin{equation}  \label{eq:relat_barT_T}
T_n(b)\leq T^{\ast}_n\left(b\right)\mbox{  under }H_0:  \beta=b\text{.}
\end{equation}
To this purpose define $Y^{\ast}_1,\ldots,Y^{\ast}_n$ by 
$$
Y^{\ast}_i=1\{U_i\geq 0\}, \quad i=1,\ldots,n\text{,}
$$
and define $T^{\ast}_n\left( b\right)\equiv \max \{ 0,T^{\ast}_u\left( b, \mathcal{V}_u\right),T^{\ast}_l\left(b, \mathcal{V}_l\right) \}$ analogously to \eqref{stat defs line 1} -- \eqref{stat defs line 2} but with
\begin{eqnarray*} 
\hat{m}_u^\ast(b,v)&\equiv&\mathbb{E}_n\left[ (2Y^{\ast}-1)1\{Xb\geq 0, Xv<0\} \right]\\
\hat{m}_l^\ast(b,v)&\equiv&\mathbb{E}_n\left[ (1-2Y^{\ast})1\{Xb\leq 0, Xv>0\} \right],
\end{eqnarray*}
replacing $\hat{m}_u(b,v)$ and $\hat{m}_l(b,v)$, respectively. The random variable $T^{\ast}_n\left( b\right)$ itself is not observed because $Y_1^{\ast},\ldots,Y_n^{\ast}$ are not observed, but the finite sample distribution of $T^{\ast}_n\left( b\right)$ given $\mathcal{X}_n$ is known since $(2Y^{\ast}_1-1,\ldots,2Y^{\ast}_n-1)$ are independent Rademacher random variables conditional on $\mathcal{X}_n$.

Thus, for a given level $\alpha\in(0,1)$, the critical value used for our test is the conditional $1-\alpha$ quantile of $T^{\ast}_n\left( b\right)$ given $\mathcal{X}_n$, namely
$$
q_{1-\alpha} \equiv \inf\{c \in \mathbb{R}: \mathbb{P}\left(T^{\ast}_n\left( b\right)\leq c\mid \mathcal{X}_n\right)\geq 1-\alpha\}\text{.}
$$
This critical value can be computed up to arbitrary accuracy by drawing a large number of simulations, each of which comprises a sequence of $n$ independent Rademacher random variables.

The relationship between $T_n(b)$ and $T^{\ast}_n\left(b\right)$ in \eqref{eq:relat_barT_T} implies Theorem \ref{Theorem: size control}, establishing finite sample size control of the proposed test.  As is the case with all formal mathematical results stated in the paper, the proofs of inequality \eqref{eq:relat_barT_T} and Theorem \ref{Theorem: size control} are in the Appendix.

\begin{theorem}
\label{Theorem: size control} Let Assumption \ref{Assumption: Model and Sampling Process} hold. Under the null $H_0: \beta=b$, $\mathbb{P}(T_n(b)\leq q_{1-\alpha}\mid \mathcal{X}_n)\geq 1-\alpha$.
\end{theorem}

Theorem \ref{Theorem: size control} establishes finite sample size control of the rejection rule $1\{ T_n(b) >q_{1-\alpha} \}$ for hypothesis test \eqref{hypothesis test - two sided}.  While it is possible that $\mathbb{P}\left(T_n(\beta)\leq q_{1-\alpha}\mid \mathcal{X}_n\right)$ strictly exceeds $1-\alpha$, the following theorem shows that a test with a smaller critical value cannot achieve size control if the critical value is a deterministic function of $\mathcal{X}_n$. It should however be noted that Theorem \ref{theorem:power} is silent with regard to critical values that are a function of both $X_1,\ldots,X_n$ and $Y_1,\ldots,Y_n$.

\begin{theorem}\label{theorem:power}
Let Assumption \ref{Assumption: Model and Sampling Process} hold. The proposed critical value is not conservative in the sense that, for every function $cv$ of  $\mathcal{X}_n$ with $cv<q_{1-\alpha}$, there is a distribution of $(U_1,\ldots,U_n)$ given $\mathcal{X}_n$ under which $\mathbb{P}(T_n(b)\leq cv\mid \mathcal{X}_n)<1-\alpha$ under the null $H_0: \beta=b$. 
\end{theorem}

Theorem \ref{theorem:power} implies that the use of a more stringent critical value that is a function of $\mathcal{X}_n$ for testing using $T_n(b)$ is not possible without losing size control. 

We can establish a power result for our test as a function of a measure of the violation of moment inequalities that define the finite sample identified set.  Specifically, Hoeffding's inequality is used to establish a lower bound on finite sample power for certain violations of the inequalities \eqref{UCon Inequality Theorem Pos} and \eqref{UCon Inequality Theorem Neg} from Theorem \ref{Theorem: ID set characterization}.  The result is given in the following theorem.

\begin{theorem}\label{theorem:power_bound}
Let Assumptions \ref{Assumption: Model and Sampling Process} and \ref{Assumption: Uindendepednece} hold, and let $\rho$ be any number in $(0,1)$. If there is $v\in\mathcal{V}_u$ such that 
\begin{multline} 
\mathbb{E}\left[\ \mathbb{E}_n\left[ (2Y-1) 1\{Xb\geq 0, Xv<0\}\right]\mid \mathcal{X}_n\right]  \\
\leq 
-\frac{1}{\sqrt{n}}\left(q_{1-\alpha}\max\left\{\epsilon,\sqrt{\frac{\mathbb{E}_n[1\{Xb\geq 0, Xv<0\}]}{1+q_{1-\alpha}^2/n}}\right\}+\sqrt{2\log(1/\rho)\mathbb{E}_n[1\{Xb\geq 0,Xv<0\}]}\right),  \label{eq:distant_alt}
\end{multline}
or there is $v\in\mathcal{V}_l$ such that 
\begin{multline}
\mathbb{E}\left[\ \mathbb{E}_n\left[ (1-2Y) 1\{Xb\leq 0, Xv>0\}\right]\mid \mathcal{X}_n\right] \\ 
\leq 
-\frac{1}{\sqrt{n}}\left(q_{1-\alpha}\max\left\{\epsilon,\sqrt{\frac{\mathbb{E}_n[1\{Xb\leq 0, Xv>0\}]}{1+q_{1-\alpha}^2/n}}\right\}+\sqrt{2\log(1/\rho)\mathbb{E}_n[1\{Xb\leq 0,Xv>0\}]}\right), \label{eq:distant_alt2}
\end{multline}
then the rejection probability is at least $1-\rho$, i.e., $\mathbb{P}(T_n(b)>q_{1-\alpha}\mid \mathcal{X}_n)\geq 1-\rho$.
\end{theorem}

Furthermore, inversion of the conditions in Theorem \ref{theorem:power_bound} provides a power guarantee for any parameter value $b \notin \mathcal{B}_n^{\ast}$. 
\begin{corollary}\label{Corollary: Power Guarantee 1}
Let Assumptions \ref{Assumption: Model and Sampling Process} and \ref{Assumption: Uindendepednece} hold. 
For any $b \notin \mathcal{B}_n^{\ast}$, the rejection probability $\mathbb{P}(T_n(b)>q_{1-\alpha}\mid \mathcal{X}_n)$ is at least the maximum of the following two expressions: 
\begin{equation}\label{power bound upper}
\max_{v\in\mathcal{V}_u}\left(
1-
\exp\left(-\frac{1}{2}\left(\max\left\{0,\sqrt{n}\tilde{\zeta}_u(b,v)-q_{1-\alpha}\max\left\{\tilde{\epsilon}_u(b,v),(1+q_{1-\alpha}^2/n)^{-1/2}\right\}\right\}\right)^2\right)
\right)\text{,}
\end{equation}
\begin{equation}\label{power bound lower}
\max_{v\in\mathcal{V}_l}\left(
1-
\exp\left(-\frac{1}{2}\left(\max\left\{0,\sqrt{n}\tilde{\zeta}_l(b,v)-q_{1-\alpha}\max\left\{\tilde{\epsilon}_l(b,v),(1+q_{1-\alpha}^2/n)^{-1/2}\right\}\right\}\right)^2\right)
\right)\text{,}
\end{equation}
where the quantities in the above expressions are defined as 
$$
\tilde{\zeta}_u(b,v)
\equiv
\frac{-\mathbb{E}\left[\ \mathbb{E}_n\left[ (2Y-1) 1\{Xb\geq 0, Xv<0\}\right]\mid \mathcal{X}_n\right]}{\sqrt{\mathbb{E}_n[1\{Xb\geq 0,Xv<0\}]}}
$$
$$
\tilde{\zeta}_u(b,v)
\equiv
\frac{-\mathbb{E}\left[\ \mathbb{E}_n\left[ (1-2Y) 1\{Xb\leq 0, Xv>0\}\right]\mid \mathcal{X}_n\right]}{\sqrt{\mathbb{E}_n[1\{Xb\leq 0,Xv>0\}]}}
$$
$$
\tilde{\epsilon}_u(b,v)
\equiv
\frac{\epsilon}{\sqrt{\mathbb{E}_n[1\{Xb\geq 0,Xv<0\}]}}
$$
$$
\tilde{\epsilon}_l(b,v)
\equiv
\frac{\epsilon}{\sqrt{\mathbb{E}_n[1\{Xb\leq 0,Xv>0\}]}}.
$$
\end{corollary}

The power bound provided by Theorem \ref{theorem:power_bound} depends on the degree to which the inequalities \eqref{UCon Inequality Theorem Pos} and \eqref{UCon Inequality Theorem Neg} that characterize $\mathcal{B}_n^{\ast}$ are violated relative to $\sqrt{\mathbb{E}_n[1\{Xb\geq 0,Xv<0\}]}$ and $\sqrt{\mathbb{E}_n[1\{Xb\leq 0,Xv>0\}]}$. However, Theorem \ref{theorem:power_bound} further implies an explicit mapping between (i) the extent to which a given parameter vector $b$ violates the inequalities that define the finite sample identified set and (ii) a lower bound on the finite sample power of our test for $\beta = b$ that does not depend on sample quantities. To see this, define
\begin{equation*}
Q(b) \equiv \max \{Q_{u}(b), Q_{l}(b)\} \text{,}
\end{equation*}
where
\begin{align*}
Q_{u}(b) &\equiv - \min\{0, \min_{v \in \mathcal{V}_u} \mathbb{E}\left[\ \mathbb{E}_n\left[ (2Y-1) 1\{Xv < 0 \leq Xb\}\right] \mid \mathcal{X}_n\right]\}\text{,} \\
Q_{l}(b) &\equiv - \min\{0, \min_{v \in \mathcal{V}_l} \mathbb{E}\left[\ \mathbb{E}_n\left[ (1-2Y) 1\{Xb\leq 0 < Xv\}\right]\mid \mathcal{X}_n\right]\}\text{.}
\end{align*}
The values of $Q_{u}(b)$ and $Q_{l}(b)$ denote the maximal violation exhibited by $b$ of the inequalities \eqref{UCon Inequality Theorem Pos} and \eqref{UCon Inequality Theorem Neg} that characterize $\mathcal{B}_n^{\ast}$ in Theorem \ref{Theorem: ID set characterization}.  Theorem \ref{theorem:power_bound} implies that for $\gamma = 1 - \rho$, our test is guaranteed to have power at least $\gamma$ whenever the measure of violation $Q(b)$ is at least $C(\gamma)$, defined by
\begin{align}
\label{C_power}C(\gamma) \equiv \frac{1}{\sqrt{n}}\left(q_{1-\alpha}\max\left\{\epsilon,(1+q_{1-\alpha}^2/n)^{-1/2}\right\}+\sqrt{-2\log(1-\gamma)}\right).
\end{align}
Inversion of this relation also provides an explicit power guarantee as a function of $Q(b)$. The following corollary to Theorem \ref{theorem:power_bound} gives the formal results.
\begin{corollary}\label{power corollary}
Let Assumptions \ref{Assumption: Model and Sampling Process} and \ref{Assumption: Uindendepednece} hold and let $\gamma \in (0,1)$.
\begin{enumerate}
\item If $Q(b) \geq C(\gamma)$, then $\mathbb{P}(T_n(b)>q_{1-\alpha}\mid \mathcal{X}_n)\geq \gamma$.
\item For any $b \notin \mathcal{B}_n^{\ast}$, the rejection probability $\mathbb{P}(T_n(b)>q_{1-\alpha}\mid \mathcal{X}_n)$ is at least 
$$
1-\exp \left(-\frac{1}{2}\left(\max\left\{0,\sqrt{n}Q(b)-q_{1-\alpha}\max\left\{\epsilon, (1+q_{1-\alpha}^2/n)^{-1/2}\right\}\right\}\right)^2 \right).
$$
\end{enumerate}
\end{corollary}
In particular, Corollary \ref{power corollary} can be used to indicate how big $Q(b)$ must be in order for Theorem \ref{theorem:power_bound} to guarantee our test has power at least $\alpha$ against a parameter value $\tilde{b}$ irrespective of sample quantities.  As we show in the next section, the likelihood ratio test of $\beta = b$ against $\beta = \tilde{b}$ can only achieve power $\alpha$ against any $\tilde{b}$.

Finally, it should be noted that the power bounds delivered by Theorem \ref{theorem:power_bound} provide power guarantees, but may not be sharp.  That is, the test may achieve higher power than the bounds guarantee.

\section{Likelihood Ratio Tests}\label{Section: likelihood ratio test}

In this section, we consider the use of likelihood ratio tests under Assumptions \ref{Assumption: Model and Sampling Process} and \ref{Assumption: Uindendepednece}.  The analysis is divided according to the type of null and alternative hypotheses considered, the key distinctions being whether each hypothesis is simple or composite.  When considering the likelihood ratio test, the relevant consideration is the set of conditional distributions for $Y_1,\ldots,Y_n$ given $\mathcal{X}_n$ allowed under each hypothesis. A point null hypothesis for $\beta$ is a composite hypothesis because it does not specify the distribution of $U_1,\ldots,U_n$ conditional on $\mathcal{X}_n$, and therefore admits many different possible conditional distributions for $Y_1,\ldots,Y_n$ even under the maintained assumptions that $\mathbb{P}\left( U_i\geq 0\mid \mathcal{X}_n\right) =1/2$ and that $\mathbb{P}\left( U_{1}\leq t_1,\ldots,U_n \leq t_n \mid \mathcal{X}_n \right) = \prod_{i=1}^n \mathbb{P}\left( U_i\leq t_i \mid \mathcal{X}_n \right)$ for any real numbers $t_1,\ldots,t_n$. 

The likelihood ratio test is most powerful as long as the alternative hypothesis is simple. However, as shown below, when the alternative hypothesis is composite, the likelihood ratio test based on the least favorable pair has only trivial power (i.e. its power is equal to its size, $\alpha$) because the least favorable null and alternative distributions coincide. Unlike the likelihood ratio test, our proposed test has power greater than its size $\alpha$ against alternatives $\tilde{b}$ that are sufficiently in violation of the moment inequalities that characterize the finite sample identified set, as can be seen from Theorem \ref{theorem:power_bound}.

Let $G_{U\mid \mathcal{X}_n}\left(\cdot \right)$ denote the conditional distribution of $U_1,\ldots,U_n$ given $\mathcal{X}_n$. Under Assumptions \ref{Assumption: Model and Sampling Process} and \ref{Assumption: Uindendepednece}, the likelihood of $(Y_1,\ldots,Y_n)=(y_1,\ldots,y_n)$ conditional on $\mathcal{X}_n$ is defined by the product form
$$
P\left(y_1,\ldots,y_n;\beta,G_{U\mid\mathcal{X}_n}\right) \equiv \prod_{i=1}^n \left( y_i \cdot G_i\left( [- X_i \beta, \infty) \right) + (1-y_i) \cdot G_i\left( (-\infty,- X_i \beta) \right) \right) \text{,}
$$
where $G_i$ denotes the marginal distribution of $U_i$ conditional on $\mathcal{X}_n$ induced by $G_{U\mid\mathcal{X}_n}$, i.e., $G_i\left(\mathcal{T}\right) \equiv G_{U\mid\mathcal{X}_n}\left( \{s \in \mathbb{R}^n: s_i \in \mathcal{T}  \}\right)$ for every measurable set $\mathcal{T} \subset \mathbb{R}$.

Below we begin by considering the simplest case, in which the null and alternative hypothesis both specify each of $\beta$ and $G_{U\mid \mathcal{X}_n}$, and hence uniquely determine the likelihood $P\left(y_1,\ldots,y_n;\beta,G_{U\mid \mathcal{X}_n}\right)$, before considering composite hypotheses.
 
\subsubsection*{Simple Null Hypotheses and Simple Alternative Hypotheses}
 
Consider the following test of two simple hypotheses.
\begin{equation}\label{hypothesis test - simple}
H_{0}: \left(\beta,G_{U\mid \mathcal{X}_n}\right) = \left(b,G\right) \quad \text{versus} \quad H_1: \left(\beta,G_{U\mid \mathcal{X}_n}\right) = \left(\tilde{b},\tilde{G}\right)\text{.}
\end{equation}
The likelihood ratio test is the test that rejects $H_0$ in favor of $H_1$ with probability $\phi(y_1,\ldots,y_n)$ defined by
$$
\phi(y_1,\ldots,y_n)=
\begin{cases}
1&\mbox{ if }P\left(y_1,\ldots,y_n;\tilde{b},\tilde{G}\right)>kP\left(y_1,\ldots,y_n;b,G\right)\\
\xi&\mbox{ if }P\left(y_1,\ldots,y_n;\tilde{b},\tilde{G}\right)=kP\left(y_1,\ldots,y_n;b,G\right)\\
0&\mbox{ otherwise}\text{,}
\end{cases}
$$
where $k$ and $\xi\in[0,1]$ are chosen such that 
\begin{equation}\label{simple test size control}
\sum_{(y_1,\ldots,y_n)}\phi(y_1,\ldots,y_n)P\left(y_1,...,y_n;b,G\right)
=\alpha\text{,}
\end{equation}
where the summation is to be understood as taken over all $(y_1,\ldots,y_n) \in \{0,1\}^n$.\footnote{\label{footnote_non-randomized}This test uses randomization in the event that $P\left(y_1,...,y_n;b',G'\right)=kP\left(y_1,...,y_n;b,G\right)$. If a non-randomized implementation of the LRT is preferred that can be done in the usual way by making the modification
\begin{equation}
\phi(y_1,\ldots,y_n)=
\begin{cases}
1&\mbox{ if }P\left(y_1,...,y_n;b',G'\right)>kP\left(y_1,...,y_n;b,G\right)\text{,}\\
0&\mbox{ otherwise}\text{,}
\end{cases}
\end{equation}
with $k$ chosen as small as possible subject to the constraint that
\begin{equation}
\sum_{(y_1,\ldots,y_n)}\phi(y_1,\ldots,y_n)P\left(y_1,...,y_n;b,G\right)
\leq \alpha \text{.}
\end{equation}
The conclusion of Theorem \ref{Theorem: likelihood ratio test most powerful with composite null} would then be that the LRT is a most powerful non-randomized test of \eqref{hypothesis test - composite v simple} subject to achieving size control. Subsequent results could be similarly modified without substantive change of conclusions.} Defining 
\begin{equation}\label{p_i definition}
p_i \equiv G_i\left( [-X_i b, \infty) \mid \mathcal{X}_n \right) \text{, }
\end{equation}
and
\begin{equation}\label{p_i tilde definition}
\tilde{p}_i \equiv \tilde{G}_i\left( [-X_i \tilde{b}, \infty) \mid \mathcal{X}_n \right) \text{,}
\end{equation}

The test $\phi(Y_1,\ldots,Y_n)$ can then be written as  
$$
\phi(y_1,\ldots,y_n)=
\begin{cases}
1&\mbox{ if }\prod_{i=1}^n\tilde{p}_i^{y_i}(1-\tilde{p}_i)^{1-y_i}>k\prod_{i=1}^n{p}_i^{y_i}(1-{p}_i)^{1-y_i}\text{,}\\
\xi&\mbox{ if }\prod_{i=1}^n\tilde{p}_i^{y_i}(1-\tilde{p}_i)^{1-y_i}=k\prod_{i=1}^n{p}_i^{y_i}(1-{p}_i)^{1-y_i}\text{,}\\
0&\mbox{ otherwise, }
\end{cases}
$$
where $k$ and $\xi\in[0,1]$ are chosen such that
\begin{equation*}
\sum_{(y_1,\ldots,y_n) \in \{0,1\}^n}\phi(y_1,\ldots,y_n)\prod_{i=1}^n{p}_i^{y_i}(1-{p}_i)^{1-y_i}=\alpha.
\end{equation*}

When the null and alternative distributions of a hypothesis test are simple, then by the Neyman-Pearson Lemma (e.g., Chapter 3.2 of \cite{Lehmann/Romano:05}), the likelihood ratio test is the uniformly most powerful test.  However, the hypothesis test of interest, \eqref{hypothesis test - two sided}, differs because neither the null or alternative specify $G_{U\mid \mathcal{X}_n}$.

\subsubsection*{Composite Null Hypothesis and Simple Alternative Hypothesis}

In a step closer to the two-sided hypothesis test \eqref{hypothesis test - two sided}, consider first the hypothesis test 
\begin{equation}\label{hypothesis test - composite v simple}
H_{0}: \beta = b \quad \text{versus} \quad H_1: \left(\beta,G_{U\mid \mathcal{X}_n}\right) = \left(\tilde{b},\tilde{G}\right)\text{,}
\end{equation}
which comprises the same null hypothesis of \eqref{hypothesis test - two sided} and the alternative hypothesis of \eqref{hypothesis test - simple}. This test features a composite null hypothesis and a simple alternative hypothesis for $\left(\beta,G_{U\mid \mathcal{X}_n}\right)$.

Using Theorem 3.8.1 of \cite{Lehmann/Romano:05}, we can reduce $H_0$ to a simple hypothesis by finding the least favorable distribution of $(Y_1,\ldots,Y_n)$ given $\mathcal{X}_n$. This entails pairing the null value $\beta = b$ with the distribution $G_{U\mid \mathcal{X}_n}$ among those satisfying the restrictions maintained in Assumptions \ref{Assumption: Model and Sampling Process} and \ref{Assumption: Uindendepednece} for which the likelihood ratio test has minimal power against $H_1$. Using the independence restriction of Assumption \ref{Assumption: Uindendepednece}, the admitted distributions of $(Y_1,\ldots,Y_n)$ given $\mathcal{X}_n$ are given by the marginal Bernoulli probabilities for each $i$, so that the hypotheses can be characterized by the implied collections of probabilities $(p_1,\ldots,p_n)$ and $(\tilde{p}_1,\ldots,\tilde{p}_n)$ coincident with \eqref{p_i definition} under the null and alternative, respectively. Thus the null hypothesis in \eqref{hypothesis test - composite v simple} can be written as 
$$
H_{0}: \left(\beta,G_{U\mid \mathcal{X}_n}\right) = \left(b,\prod_{i=1}^nG_i\right)\text{ for some }(G_1,\ldots,G_n).
$$

For each observation $i$ consider the probabilities $p_i$ and $\tilde{p}_i$. Note that under Assumptions \ref{Assumption: Model and Sampling Process} and \ref{Assumption: Uindendepednece}, $p_i - 1/2$ is nonpositive (nonnegative) if $X_i b$ is nonpositive (nonnegative), and likewise $\tilde{p}_i - 1/2$ is nonpositive (nonnegative) if $X_i \tilde{b}$ is nonpositive (nonnegative). When $X_i b$ and $\tilde{p}_i - 1/2$ have the same sign, then there exists a distribution of unobservable $U_i$ conditional on $\mathcal{X}_n$ such that $p_i$ can be made equal to $\tilde{p}_i$. For such $i$, the least favorable distribution will thus have $p_i = \tilde{p}_i$.  When instead $X_i b$ and $\tilde{p}_i - 1/2$ have the opposite sign, then in general $p_i$ cannot be equal to $\tilde{p}_i$ under the null due to the requirement in Assumption \ref{Assumption: Model and Sampling Process}(iii) that $\mathbb{P}(U_i \geq 0 \mid \mathcal{X}_n) = 1/2$. The probability $p_i$ will however be made as close as possible to $\tilde{p}_i$ while adhering to this assumption if the conditional distribution of $U_i$ allocates all mass to regions in which $| U_i | > | X_i b |$, so that $p_i =1/2$.

We now formally define $\bar{p}_i$ to denote the null probabilities $p_i = \mathbb{P}\{Y_i = 1\mid \mathcal{X}_n\}$ in keeping with the above intuition for constructing the least favorable distribution under the composite null $\beta = b$ as follows:
\begin{equation}\label{p_i_bar_definition}
\bar{p}_i=
\begin{cases}
1/2&\mbox{ if }X_ib (\tilde{p}_i - 1/2) <0 \mbox{ or }(X_ib = 0 \text{ and } \tilde{p}_i < 1/2 )\text{,}\\
\tilde{p}_i&\mbox{ otherwise}\text{.}
\end{cases}
\end{equation}
These probabilities lie within the null set in \eqref{hypothesis test - composite v simple} because $p_i = 1/2$ is possible for any value of $b$, and $p_i = \tilde{p}_i$ is possible when $\tilde{p}_i - 1/2$ and $X_ib$ are either both positive or negative, as well as when $X_ib = 0$ and $\tilde{p}_i \geq 1/2$. By the Neyman-Pearson Lemma, the likelihood ratio test for the simple null of $\bar{p}_1,\ldots,\bar{p}_n$ against $\tilde{p}_1,\ldots,\tilde{p}_n$ is most powerful. Theorem \ref{Theorem: likelihood ratio test most powerful with composite null} verifies that the likelihood ratio test for this simple hypothesis controls size under the composite null $\beta = b$, thus enabling application of Theorem 3.8.1 of \cite{Lehmann/Romano:05} and establishing that these null probabilities are least favorable.
   
To proceed consider the corresponding test $\bar{\phi}(Y_1,\ldots,Y_n)$ for \eqref{hypothesis test - composite v simple} that makes use of this configuration of implied probabilities under the null. 
\begin{equation}\label{likelihood ratio test with composite null}
\bar{\phi}(y_1,\ldots,y_n)=
\begin{cases}
1&\mbox{ if }\prod_{i=1}^n\tilde{p}_i^{y_i}(1-\tilde{p}_i)^{1-y_i}>k\prod_{i=1}^n\bar{p}_i^{y_i}(1-\bar{p}_i)^{1-y_i}\text{,}\\
\xi&\mbox{ if }\prod_{i=1}^n\tilde{p}_i^{y_i}(1-\tilde{p}_i)^{1-y_i}=k\prod_{i=1}^n\bar{p}_i^{y_i}(1-\bar{p}_i)^{1-y_i}\text{,}\\
0&\mbox{ otherwise, }
\end{cases}
\end{equation}
where $k$ and $\xi\in[0,1]$ are chosen such that 
\begin{equation}\label{size control composite alternative}
\sum_{(y_1,\ldots,y_n) \in \{0,1\}^n}\bar{\phi}(y_1,\ldots,y_n)\prod_{i=1}^n\bar{p}_i^{y_i}(1-\bar{p}_i)^{1-y_i}=\alpha.
\end{equation}
There is the following result, leveraging Theorem 3.8.1 of \cite{Lehmann/Romano:05}.

\begin{theorem}\label{Theorem: likelihood ratio test most powerful with composite null}
Let Assumptions \ref{Assumption: Model and Sampling Process} and \ref{Assumption: Uindendepednece} hold. Consider the null and alternative hypotheses stated in \eqref{hypothesis test - composite v simple} and the test $\bar{\phi}(y_1,\ldots,y_n)$ defined in \eqref{likelihood ratio test with composite null} and \eqref{size control composite alternative}. If the null hypothesis is true then the rejection probability is no greater than $\alpha$, that is, $\mathbb{E}\left[\bar{\phi}(Y_1,\ldots,Y_n)\mid \mathcal{X}_n\right]\leq\alpha$.
Moreover,
\begin{equation}\label{LFD under composite null}
P\left(y_1,\ldots,y_n;b,G_{U\mid \mathcal{X}_n}\right) = \prod_{i=1}^n\bar{p}_i^{y_i}(1-\bar{p}_i)^{1-y_i}
\end{equation}
is the least favorable distribution of $\left(Y_1,\ldots,Y_n\right)$ given $\mathcal{X}_n$ under $H_0$ against $H_1$, and the test $\bar{\phi}(Y_1,\ldots,Y_n)$ is a most powerful test of $H_0$ against $H_1$.
\end{theorem}

\subsubsection*{Composite Null Hypothesis and Composite Alternative Hypothesis}

When the researcher's goal is to test $\beta = b$ against the simple alternative hypothesis in \eqref{hypothesis test - composite v simple} that completely specifies the conditional distribution of $Y_1,\ldots,Y_n$ given $\mathcal{X}_n$, then the likelihood ratio test implemented by adopting the rejection probability specified in \eqref{p_i_bar_definition} -- \eqref{size control composite alternative} is most powerful. When instead the researcher wishes to control power against a composite alternative, such as that of
\begin{equation}\label{hypothesis test - composite composite}
H_{0}: \beta = b \quad \text{versus} \quad H_1:  \beta = \tilde{b}, 
\end{equation}
Theorem \ref{Theorem: likelihood ratio test most powerful with composite null} is silent because each conditional distribution of $Y_1,\ldots,Y_n$ allowed under the alternative hypothesis will result in a different likelihood ratio test. However, following arguments in Chapter 8.1 of \cite{Lehmann/Romano:05}, it is straightforward to construct the least favorable pair of distributions for this test. Note that for \textit{any} $\tilde{b}$ hypothesized under the alternative, distributions $\tilde{G}_i$ can be specified such that for all $i$:
\begin{equation}
\tilde{p}_i \equiv \tilde{G}_i\left( [-X_i \tilde{b}, \infty) \mid \mathcal{X}_n \right) = 1/2\text{.}
\end{equation}
In particular, this is achieved by distributions $\tilde{G}_i$ that allocate all mass to regions on which $|U_i| \geq |X_i\tilde{b}|$, while obeying the constraint $\mathbb{P}\left( U_{i}\geq 0\mid \mathcal{X}_n\right) =1/2$.\footnote{If the distribution of $U_i$ were restricted to have positive density on $\mathbb{R}$, then $G_i$ could be specified to allocate probability $1-\epsilon$ to $\{u_i: |u_i| \geq |X_i b^{\prime}| \}$ for any small $\epsilon > 0$.} Such a combination of $\left(\tilde{b},\tilde{G}\right)$ under the alternative yields $\tilde{p}_i = 1/2$ for all $i$. This conclusion holds irrespective of the hypothesized value of $\tilde{b}$ in \eqref{hypothesis test - composite v simple}.  Indeed it also holds for the value of $\beta = b$ hypothesized under the null. Correspondingly the least favorable $\{\bar{p}_i:i=1,...,n\}$ under the null given by \eqref{p_i_bar_definition} when all $\tilde{p}_i = 1/2$ is given by $\bar{p}_i = 1/2$ for all $i$. Thus from Theorem \ref{Theorem: likelihood ratio test most powerful with composite null} we have the following implications for the two-sided test of $\beta = b$ against $\beta =\tilde{b}$.

\begin{corollary}\label{corollary:trivial_LRT}
Let Assumptions \ref{Assumption: Model and Sampling Process} and \ref{Assumption: Uindendepednece} hold. The least favorable pair of $\{\bar{p}_i:i=1,...,n\}$ and $\{\tilde{p}_i:i=1,...,n\}$ for the test \eqref{hypothesis test - composite v simple} is given by $\bar{p}_i = 1/2$ for all $i$ and $\tilde{p}_i = 1/2$ for all $i$. Moreover, since the conclusion holds for any $b$ and $\tilde{b}$, this is also the least favorable pair for the two-sided hypothesis test of $\beta = b$ in \eqref{hypothesis test - two sided}.
\end{corollary}

This corollary is a direct implication of Theorem \ref{Theorem: likelihood ratio test most powerful with composite null}.  While it may be conceptually appealing to construct the likelihood ratio test for the two-sided test (\ref{hypothesis test - composite composite}) based on the least favorable pair, we see from \eqref{likelihood ratio test with composite null} and \eqref{size control composite alternative} that this results in a test which rejects the null hypothesis with probability $\alpha$ irrespective of the data. This is because the hypothesis $\beta=b$ always includes the distribution under which ${p}_i = 1/2$ for all $i$, for any hypothesized value of $b$. Thus the sets of feasible conditional distributions for $\left(Y_1,...,Y_n\right)$ given $\mathcal{X}_n$ compatible with each of the two hypotheses, $\beta=b$ and $\beta=\tilde{b}$ overlap. Therefore, in the language of  \cite{Chesher/Rosen:17GIV} and \cite{Kaido/Zhang:19}, even if $\tilde{b}$ lies outside the finite sample identified set, parameter values $b$ and $\tilde{b}$ are potentially observationally equivalent.

Against certain alternatives, the likelihood ratio test using the least favorable pair can achieve higher power than the inequality test using $T_n(b)$, but the reverse is also true, so that neither test dominates. Unlike the likelihood ratio test for the composite hypothesis \ref{hypothesis test - two sided}, the inequality test we propose in this paper based on $T_n(b)$ does not simply reject with probability $\alpha$ irrespective of the data.  Moreover, Theorem \ref{theorem:power_bound} and Corollaries \ref{power corollary} and \ref{corollary:trivial_LRT} provide lower power envelopes as a function of the violation of the inequalities characterizing the finite sample identified set.  In particular, these results can be used to characterize values of $\tilde{b}$ against which the inequality test achieves nontrivial power by setting $\gamma > \alpha$. For such alternatives, the inequality test achieves a higher power than the likelihood ratio test.

\section{Monte Carlo Experiments}\label{Section: Monte Carlos}

In this section, we present Monte Carlo results illustrating the relative performance of our test compared to that of a test using the smoothed maximum score estimator from \cite{horowitz1992smoothed} with a bootstrap critical value following \cite{horowitz:2002}, and a likelihood ratio test against a simple alternative hypothesis.\footnote{For the likelihood ratio test, as with our own test, we use a non-randomized test, so that the size of the test is as close as possible to $\alpha$ without exceeding it, as described in Footnote \ref{footnote_non-randomized}.} We consider examples in which there are two covariates. In each design, the first component of $\beta$ is normalized to one, and we report the results of conducting tests of the null hypothesis $\beta=(1,\theta)^{\prime}$ against the alternative hypothesis $\beta \neq (1,\theta)^{\prime}$ for $\theta$ ranging from $\theta-3$ to $\theta+3$. In the population design for the actual data generation process, the true parameter value is $\beta = (1,1)^{\prime}$.

For the likelihood ratio test, the alternative hypothesis consists of the simple hypothesis corresponding to the true data distribution. This is known in the Monte Carlo simulation, but would be unknown in practice. This likelihood ratio test is therefore the (infeasible) optimal test of the hypothesis $\beta=b$ against $\beta = (1,1)^{\prime}$ paired with $G_{U\mid \mathcal{X}_n}$, producing Bernoulli probabilities for $\tilde{p}_1,\ldots,\tilde{p}_n$ for the actual DGP employed in the Monte Carlo simulations. This is thus a test of the composite null $\beta = b$ against a point alternative as described by \eqref{hypothesis test - composite v simple}.

In the Monte Carlo experiments, we use one random draw of the $n$-tuple $(X_1,\ldots,X_n)$, and 500 independent draws of the $n$-tuple $(Y_1,\ldots,Y_n)$, where the sample size is $n=100$. In this simulation design the finite sample identified set is thus fixed across experiments.

To implement the three inference methods, we use the following tuning parameters. The variable $\epsilon$ was set to MATLAB's \textsf{eps} value of approximately $2.2 \cdot 10^{-16}$. To compute the critical value, we use 500 random draws of $n$ Rademacher random variables and 500 samples for the bootstrap procedure described in \cite{horowitz:2002}.  In all cases we considered tests with size $\alpha = 0.10$.

Following \cite[Section 3]{horowitz1992smoothed}, $X=(X_1,X_2)$ is a bivariate normal random vector with $\mathbb{E}[X_1]=0$, $\mathbb{E}[X_2]=1$, $Var(X_1)=Var(X_2)=1$ and $Cov(X_1,X_2)=0$. The distribution of unobservable $U$ in each design is as follows, with $U$ independent of $X$ in designs 1-3 and $V$ independent of $X$ in design 4.
\begin{itemize}
\item Design 1: $U$ is distributed according to the logistic distribution
with mean zero and variance one. 
\item Design 2: $U$ is uniformly distributed on $[-\sqrt{12}/2,\sqrt{12}/2]$.
\item Design 3: $U$ is distributed according to the Student's $t$ distribution
normalized to have variance one. 
\item Design 4: $U=0.25(1+2Z^2+Z^4)V$ where $Z=X_1+X_2$ and $V$ is distributed
according to the logistic distribution with mean zero and variance one.
\end{itemize}

Figure \ref{fig_MC} presents non-rejection frequencies for $\theta$ ranging from $\theta-3$ to $\theta+3$ for $\alpha = 0.10$. For the non-rejection frequencies, we conduct inference using the smoothed maximum score bootstrap implementation described in Horowitz (2002) (Horowitz-Bootstrap), our proposed method (Rosen-Ura), and the (infeasible) likelihood ratio test (LRT). For reference, the finite sample identified set (FSID) is also illustrated with height of $1-\alpha$ on the vertical axis by a dashed blue line. The qualitative comparison among the three inference methods is similar for all four designs. The non-rejection frequency for our method is maximized around the population value of $\theta=1$ and exceeds $0.90$ in a range spanning from about 0.6 to 1.7. Our method is always less powerful than the infeasible likelihood ratio test, which results from the fact that the alternative hypothesis for our method is the composite hypothesis $\beta \neq b$, while the likelihood ratio test is based on using the actual population distribution as the simple alternative hypothesis. This is what makes the likelihood ratio test considered here infeasible; in an application, the population distribution of $Y_1,\ldots,Y_n$ conditional on $\mathcal{X}_n$ is unknown.
    
Compared to inference using the bootstrap implementation of the smoothed maximum score estimator, we find an asymmetry in the simulation results.  For values smaller than the true parameter value, our method is more powerful than the smoothed maximum score bootstrap procedure, although the performance is almost the same in Design 4. For values larger than the true parameter value, our method is less powerful than the smoothed maximum score bootstrap procedure, but the smoothed maximum score bootstrap procedure is more powerful than the optimal test (the likelihood ratio test) for some values larger than the true parameter value. This is likely because the formal asymptotic theory for the smoothed maximum score estimator in \cite{horowitz1992smoothed} and its use for inference by way of the bootstrap implementation described in \cite{horowitz:2002} invoke stronger assumptions than those of Assumptions \ref{Assumption: Model and Sampling Process} and \ref{Assumption: Uindendepednece} in this paper, and these additional assumptions hold in the data generating processes used in these illustrations.  Further, we see in Designs 1 and 2 that inference using the smoothed maximum score bootstrap results in rejection frequencies in excess of $\alpha = 0.1$ for certain parameter values that lie inside the finite sample identified set.  This is not surprising because its performance is guaranteed to achieve size control asymptotically under conditions whereby the identified set is in fact a singleton.

In Appendix \ref{sec:additional_results}, we provide additional illustrations of simulation results for a different population parameter value $\theta_0=2$ and/or a larger sample size $n=250$. The results are qualitatively similar to those shown in Figure \ref{fig_MC}.

\section{Conclusion}\label{Section: Conclusion}

In this paper we have proposed an approach to conduct finite sample inference on the parameters of Manski's (1985) semiparametric binary response model, for which the maximum score estimator has been shown to be cube-root consistent with a non-normal asymptotic distribution when there is point identification. Our finite sample inference approach circumvents the need to accommodate the complicated asymptotic behavior of this point estimator.  Since our goal was finite sample inference, we considered the problem of making inference conditional on the $n$ covariate vectors observable in a finite sample.  With covariates taking on only a finite number of observed values, the parameter vector $\beta$ is not point identified.  We therefore employed moment inequality implications for $\beta$ for the sake of constructing our test statistic for inference, as the moment inequalities are valid no matter whether $\beta$ is point identified. In order to exposit what observable implications can be distilled on only the basis of exogenous variables observed in the finite sample, we defined the notion of a finite sample identified set.  We showed how to make use of the full set of observable implications conditional on the size $n$ sequence of exogenous variables in our construction of a test statistic $T_n(b)$.  Finite sample valid critical values were established, and were shown to be easily computed by making use of many simulations of size $n$ sequences of independent Rademacher variables. A finite sample power (lower) bound was also presented and the results of some Monte Carlo experiments were reported, illustrating the performance of the test.

Several interesting directions for future research are possible.  The maximum score estimator is one of several estimators in the econometrics literature that consistently estimate a model parameter that may only be identified under support conditions that can never be satisfied by an empirical distribution based on a finite sample.  Some such estimators, like the maximum score estimator, exhibit slower than $n^{-1/2}$ convergence rates.  Other such estimators, such as the maximum rank correlation estimator of \cite{Han:87} achieve the parametric rate.  In this paper we have exploited the particular structure of the semi-parametric binary response model, but there may nonetheless be potential to extend ideas in this paper to such settings to alleviate dependence on conditions not satisfied by empirical distributions that result from finite data.

Other possible avenues of investigation pertain to optimal testing. One direction could be to exploit the likelihood ratio test analysis in this paper to consider minimax testing rates under additional assumptions on the distribution of $U_i$.\footnote{This line of research was suggested by Tim Armstrong.} Minimax optimal estimation has recently been investigated in a setting in which covariates have high-dimension by \cite{mukherjee2019nonstandard} in an asymptotic framework under sufficient conditions for point identification. Investigation of minimax optimal estimators and tests could be interesting to consider when conditions for point identification are not guaranteed to hold. More generally, in future work we aim to continue to explore the interplay between partial identification and testability, and in particular the implications of not having point identification based on an underlying discrete data distribution, as one always has when using the empirical distribution obtained in a finite data set.

\newpage
\begin{figure}[ht]
\parbox{.4\textwidth}{
\raggedleft
\includegraphics[width=.45\textwidth,keepaspectratio]{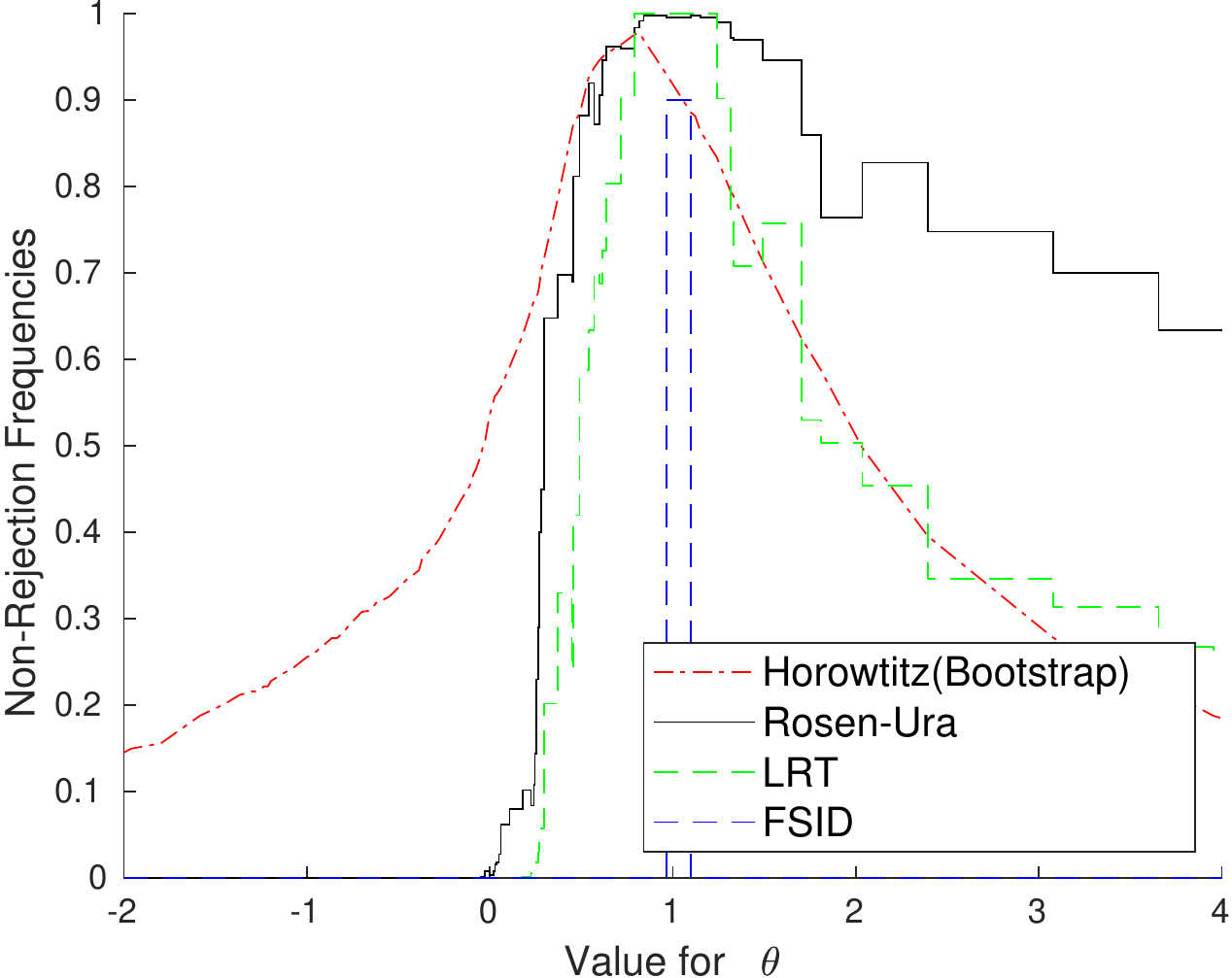}
\caption*{Figure \ref{fig_MC}.a:  Design 1.}
}
\hspace{.05\textwidth}
\parbox{.4\textwidth}{
\raggedright
\includegraphics[width=.45\textwidth,keepaspectratio]{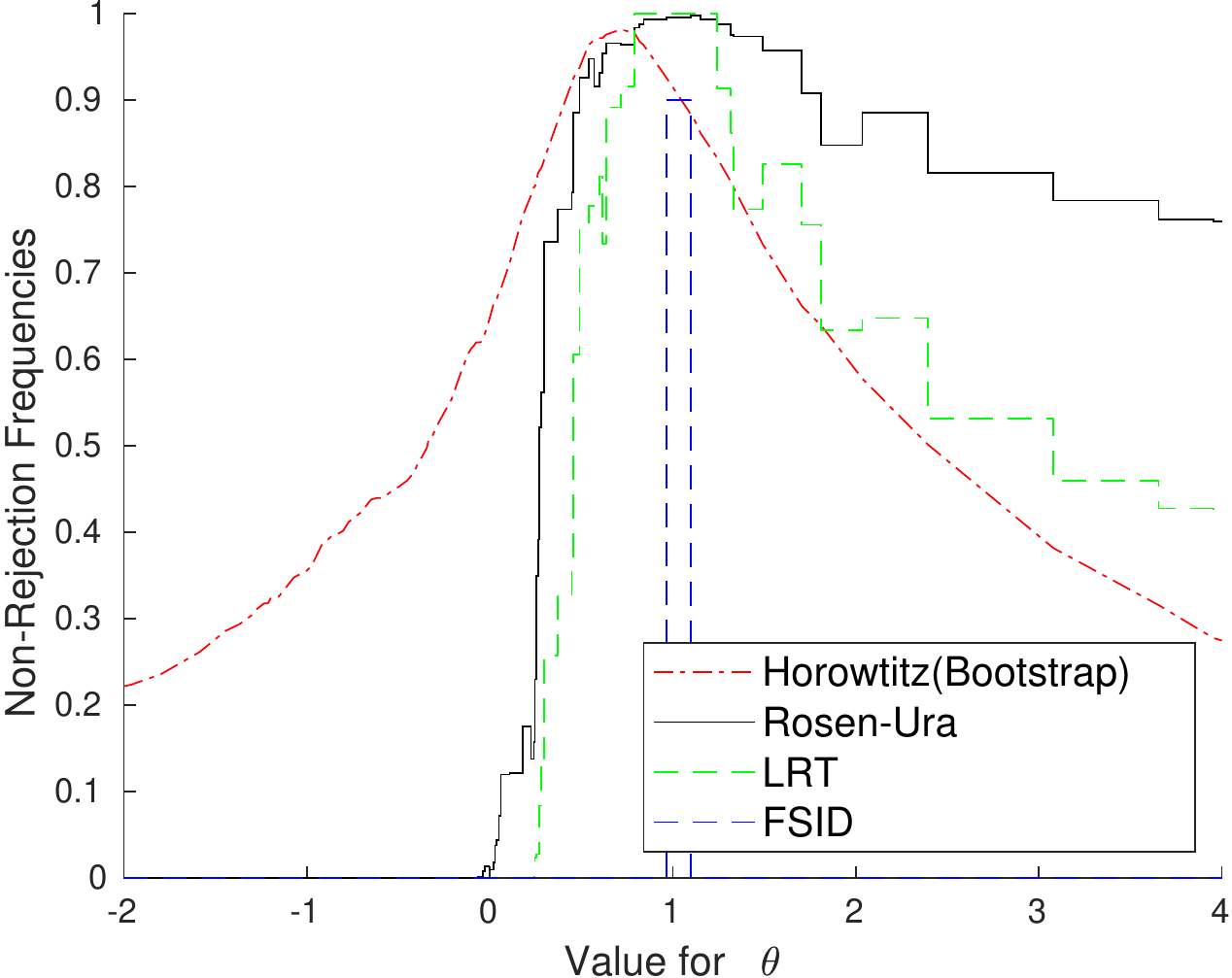}
\caption*{Figure \ref{fig_MC}.b:  Design 2.}
}
\\
\bigskip
\\
\parbox{.4\textwidth}{
\raggedleft
\includegraphics[width=.45\textwidth,keepaspectratio]{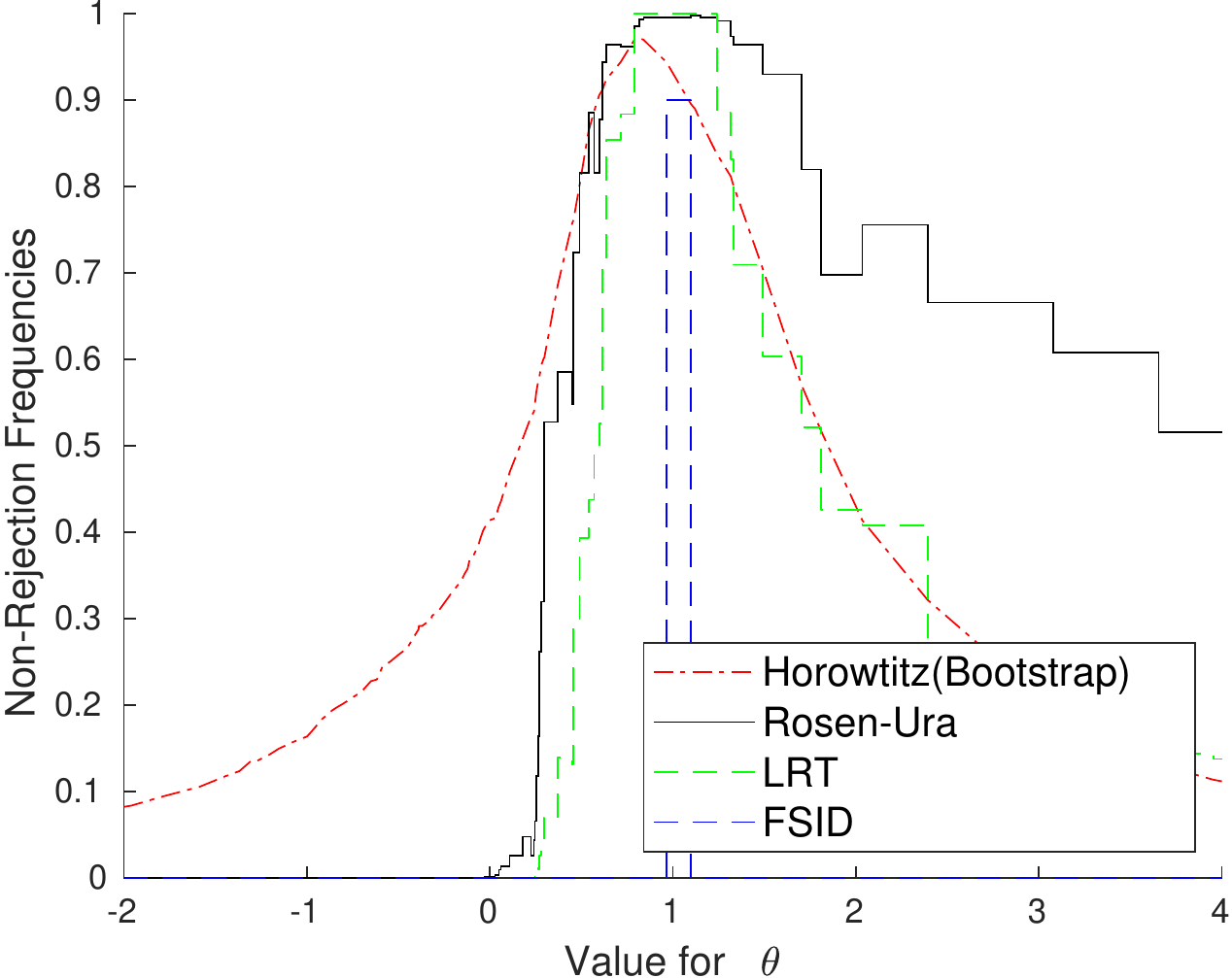}
\caption*{Figure \ref{fig_MC}.c:  Design 3.}
}
\hspace{.05\textwidth}
\parbox{.4\textwidth}{
\raggedright
\includegraphics[width=.45\textwidth,keepaspectratio]{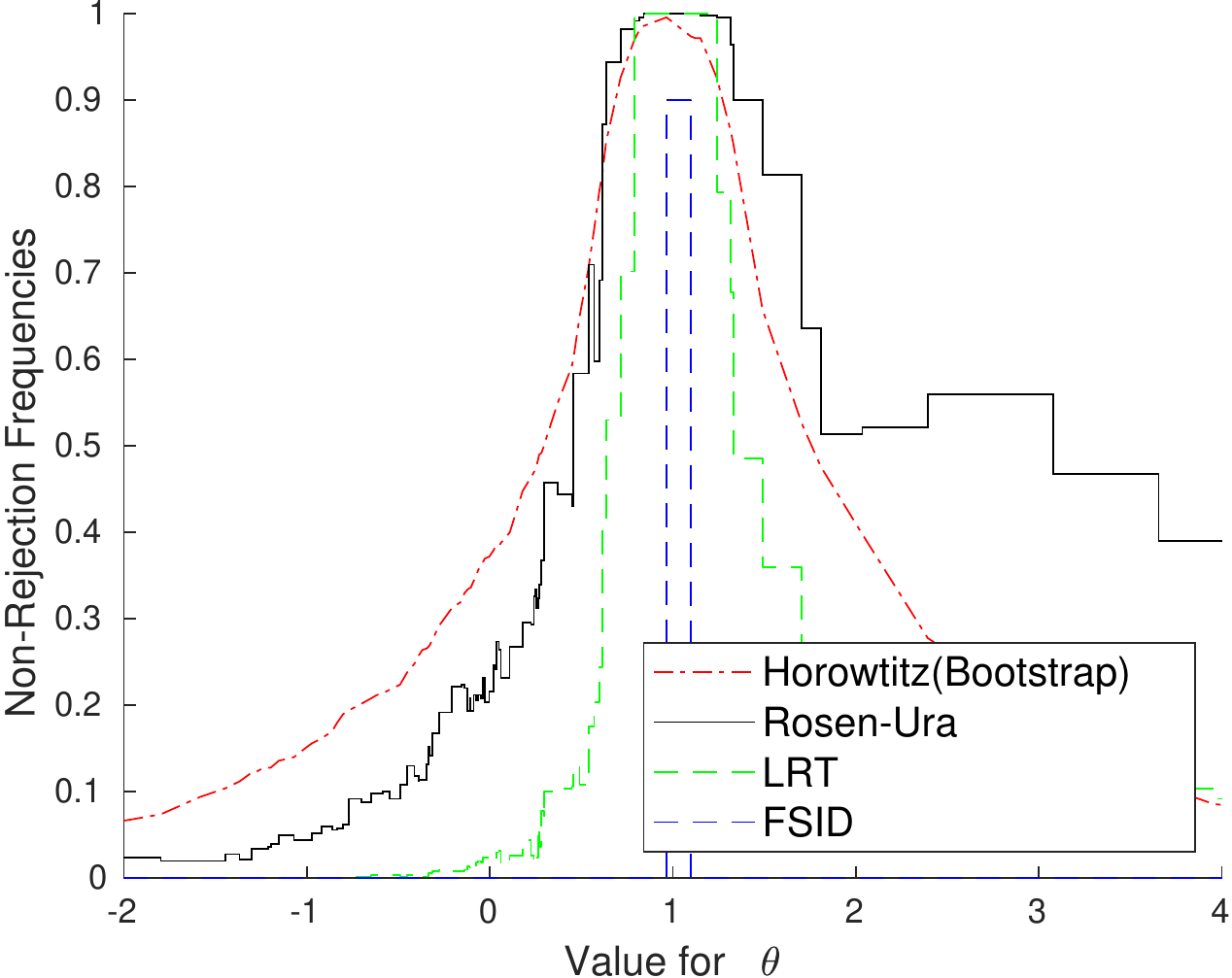}
\caption*{Figure \ref{fig_MC}.d:  Design 4.}
}
\\
\\
\caption{Non-rejection frequencies with $1-\protect\alpha=90\%$ with true $\theta_0=1$ and $n=100$. }
\label{fig_MC}
\end{figure}

\newpage
\singlespacing
\bibliographystyle{econometrica}
\bibliography{MaxScore}
\onehalfspacing
\newpage

\appendix 
\section{Proofs} 

\begin{proof}[Proof of Lemma \ref{Lemma:Conditional Moment Inequalities}]
If $X_i\beta\geq 0$, then $Y_i = 1\{ X_i\beta +U_i\geq 0\} \geq 1\{U_i\geq 0\} $ and therefore $\mathbb{E}[2Y_i-1\mid \mathcal{X}_n] \geq 2\mathbb{P}(U_i\geq 0\mid\mathcal{X}_n)-1 = 0$. If $X_i\beta\leq 0$, then $Y_i = 1\{ X_i\beta +U_i\geq 0\} \leq 1\{U_i\geq 0\} $ and therefore  $\mathbb{E}[2Y_i-1\mid \mathcal{X}_n] \leq 2\mathbb{P}(U_i\geq 0\mid\mathcal{X}_n)-1 = 0$. 
\end{proof}

\begin{proof}[Proof of Theorem \ref{Theorem: Finite Sample Identified Set}]
Directly from Lemma \ref{Lemma:Conditional Moment Inequalities}, $b\in \mathcal{B}_{n}^{\ast }$ implies 
$$
\mathbb{E}\left[ \left( 2Y_i-1\right) 1\{X_ib\geq 0\} \mid \mathcal{X}_n\right] \geq 0
\mbox{ and  }
\mathbb{E}\left[ \left( 1-2Y_i\right) 1\{X_ib\leq 0\} \mid \mathcal{X}_n\right] \geq 0.
$$ 
To demonstrate the other direction, let $b$ be any element of $\mathcal{B}$ such that $\mathbb{E}\left[ \left( 2Y_i-1\right) 1\{X_ib\geq 0\} \mid \mathcal{X}_n\right] \geq 0$ and $\mathbb{E}\left[ \left( 1-2Y_i\right) 1\{X_ib\leq 0\} \mid \mathcal{X}_n\right] \geq 0$ for every $i=1,\ldots,n$. Then
\begin{eqnarray}
\label{ineq1}X_i b \geq 0\implies\mathbb{P}(Y_i=1\mid \mathcal{X}_n) \geq 1/2\text{,}\\
\label{ineq2}X_i b \leq 0\implies\mathbb{P}(Y_i=1\mid \mathcal{X}_n) \leq 1/2\text{.}
\end{eqnarray}
To show that $b$ is in $\mathcal{B}_n^{\ast}$ as defined in Definition \ref{Def: Finite Sample Identified Set}, we now construct a collection of random variables $\{\tilde{U}_{i} :i=1,\ldots,n\} $ such that for all $i=1,\ldots,n$: (i) $\mathbb{P}(Y_i=1\{ X_i b + \tilde{U}_i \geq 0\}\mid \mathcal{X}_n)=1$, and (ii) $\mathbb{P}(\forall i, 1\{U_i \geq 0\}=1\{\tilde{U}_i \geq 0\}\mid \mathcal{X}_n) = 1$. To do so, let $\kappa_i: i=1,\ldots,n$ be $n$ positive random variables defined on $\left( \Omega ,\mathfrak{F} ,\mathbb{P}\right)$  and consider each of the cases $X_i \beta < 0$, $X_i \beta =0$, and $X_i \beta > 0$ in turn as follows.

\noindent\textbf{Case 1: $X_i \beta < 0$}.  By Lemma \ref{Lemma:Conditional Moment Inequalities},
\begin{equation}\label{C1 inequality}
\mathbb{E}[2 Y_i - 1 \mid \mathcal{X}_n] \leq 0\text{.}
\end{equation}
Let
$$
\tilde{U}_i \equiv 
\begin{cases}
\max\{-X_i b,0\} + \kappa_i&\mbox{ if }U_i \geq -X_i \beta\\
0&\mbox{ if }0 \leq U_i < -X_i \beta\\
\min\{-X_i b,0\}-\kappa_i&\mbox{ if }U_i < 0.
\end{cases}
$$
Then $1\{\tilde{U}_i \geq 0\}=1\{U_i \geq 0\}$, which verifies (ii). To verify (i), note that: 
$$
1\{ X_i b + \tilde{U}_i \geq 0\}
=1\{X_i \beta + U_i \geq 0 \}+1\{0 \leq U_i < -X_i \beta,\  X_i b\geq 0\}
=Y_i+1\{0 \leq U_i < -X_i \beta, X_i b\geq 0\}, 
$$
because 
$$
X_i b + \tilde{U}_i
=
\begin{cases}
\max\{X_i b,0\} + \kappa_i&\mbox{ if }U_i \geq -X_i \beta\\
X_i b&\mbox{ if }0 \leq U_i < -X_i \beta\\
\min\{X_i b,0\}-\kappa_i&\mbox{ if }U_i < 0.
\end{cases}
$$
Therefore, 
\begin{eqnarray*}
\mathbb{P}(Y_i=1\{ X_i b + \tilde{U}_i \geq 0\}\mid \mathcal{X}_n)
&=&
\mathbb{P}(1\{0 \leq U_i < -X_i \beta, X_i b\geq 0\}=0\mid \mathcal{X}_n)\\ 
&\geq&
\mathbb{P}(1\{0 \leq U_i < -X_i \beta\}=0\mid X_i b\geq 0,\mathcal{X}_n)\\
&=&
\mathbb{P}(U_i\geq -X_i \beta\mid X_i b\geq 0,\mathcal{X}_n)-\mathbb{P}(U_i\geq 0\mid X_i b\geq 0,\mathcal{X}_n)\\
&=&
\mathbb{P}(Y_i=1\mid X_i b\geq 0,\mathcal{X}_n)-1/2.
\end{eqnarray*}
Since \eqref{ineq1} and \eqref{C1 inequality} imply that $\mathbb{P}(Y_i=1\mid X_i b\geq 0,\mathcal{X}_n) = 1/2$, we have $\mathbb{P}(Y_i=1\{ X_i b + \tilde{U}_i \geq 0\}\mid \mathcal{X}_n)=1$, which verifies (i).


\noindent\textbf{Case 2: $X_i \beta = 0$}. Let
\begin{equation} \label{U tilde case 2}
\tilde{U}_i \equiv
\begin{cases}
\max\{-X_i b,0\} + \kappa_i&\mbox{ if }U_i \geq 0\\
\min\{-X_i b,0\} -\kappa_i&\mbox{ if }U_i < 0.
\end{cases}
\end{equation}
Then $1\{\tilde{U}_i \geq 0\}=1\{U_i \geq 0\}$, which verifies (ii). It further follows from \eqref{U tilde case 2} that
$$
X_i b + \tilde{U}_i = 
\begin{cases}
\max\{0, X_i b\} + \kappa_i &\mbox{ if }U_i \geq 0\\
\min\{0, X_i b\} - \kappa_i&\mbox{ if }U_i < 0.
\end{cases}
$$
Consequently since $X_i \beta = 0$, $X_i b+\tilde{U}_i \geq 0$ if and only if $X_i \beta+U_i \geq 0$, verifying (i).

\noindent\textbf{Case 3: $X_i \beta > 0$}. The proof is similar to Case 1 with 
$$
\tilde{U}_i \equiv 
\begin{cases}
\max\{-X_i b,0\} + \kappa_i&\mbox{ if }U_i \geq 0\\
-X_i b&\mbox{ if }-X_i \beta \leq U_i < 0\\
\min\{-X_i b,0\}-\kappa_i&\mbox{ if }U_i < -X_i b.
\end{cases}
$$
\end{proof}

\begin{proof}[Proof of Theorem \ref{Theorem: ID set characterization}]
That $b\in \mathcal{B}_{n}^{\ast }$ implies (\ref{UCon Inequality Theorem Pos}) and (\ref{UCon Inequality Theorem Neg}) is immediate. To demonstrate the other direction, we are going to show that \eqref{UCon Inequality Theorem Pos} and \eqref{UCon Inequality Theorem Neg} imply that
\begin{equation}  \label{eq:theorem1_goal}
\mathbb{E}\left[ \left( 2Y_i-1\right) 1\{X_ib\geq 0\} \mid \mathcal{X}_n\right]\geq 0\mbox{ and }\mathbb{E}\left[ \left( 1-2Y_i\right) 1\{X_ib\leq 0\} \mid \mathcal{X}_n\right]\geq 0
\end{equation}
for every $i=1,\ldots,n$. To show the result, let $b \in \mathcal{B}$ such that \eqref{UCon Inequality Theorem Pos} and \eqref{UCon Inequality Theorem Neg} hold. Let $v_u\in \mathcal{V}_u$ and $v_l\in \mathcal{V}_l$ such that $r_u\left(v_u\right) = r_u\left(\beta \right)$ and $r_l\left(v_l\right) = r_l\left(\beta \right)$. Note that under Assumption \ref{Assumption: Model and Sampling Process}\ (iii), such $v_u$ and $v_l$ exist. Lemma \ref{Lemma:Conditional Moment Inequalities} implies 
\begin{eqnarray}
-\mathbb{E}_n\left[ |\mathbb{E}\left[2Y-1\mid \mathcal{X}_n\right]|1\{Xb\geq 0, X\beta<0\}\right] 
&=& 
\mathbb{E}_n\left[ \mathbb{E}\left[2Y-1\mid \mathcal{X}_n\right]1\{Xb\geq 0, X\beta<0\}\right]  \notag \\
&=& 
\mathbb{E}_n\left[ \mathbb{E}\left[2Y-1\mid \mathcal{X}_n\right]1\{Xb\geq 0, Xv_u< 0\}\right]  \notag \\
&=& 
\mathbb{E}\left[\ \mathbb{E}_n\left[ (2Y-1) 1\{Xb\geq 0, Xv_u< 0\}\right]\mid \mathcal{X}_n\right]  \notag \\
&\geq& 
0  \label{theorem1_proof_eq1} \\
 -\mathbb{E}_n\left[ |\mathbb{E}\left[2Y-1\mid \mathcal{X}_n\right]|1\{Xb\leq 0, X\beta>0\}\right] 
&=& 
\mathbb{E}_n\left[ \mathbb{E}\left[1-2Y\mid \mathcal{X}_n\right]1\{Xb\leq 0, X\beta>0\}\right]  \notag \\
&=& 
\mathbb{E}_n\left[\mathbb{E}\left[1-2Y\mid \mathcal{X}_n\right]1\{Xb\leq 0, Xv_l>0\}\right]  \notag \\
&=& 
\mathbb{E}\left[\ \mathbb{E}_n\left[ \left(1-2Y\right) 1\{Xb\leq 0, Xv_l>0\}\right]\mid \mathcal{X}_n\right]  \notag \\
&\geq& 
0.  \label{theorem1_proof_eq2}
\end{eqnarray}
Moreover, since both $- |\mathbb{E}\left[2Y_i-1\mid \mathcal{X}_n\right]|1\{X_ib\geq 0,X_i\beta<0\}$ and $- |\mathbb{E}\left[2Y_i-1\mid \mathcal{X}_n\right]|1\{X_ib\leq 0,X_i\beta>0\}$ must be non-positive for every $i = 1,\ldots,n$, we have 
\begin{equation}  \label{theorem1_proof_eq_main}
|\mathbb{E}\left[2Y_i-1\mid \mathcal{X}_n\right]|1\{X_ib\geq 0,X_i\beta<0\}=|\mathbb{E}\left[2Y_i-1\mid \mathcal{X}_n\right]|1\{X_ib\leq 0,X_i\beta>0\}=0
\end{equation}
for every $i = 1,\ldots,n$. We can demonstrate Eq. (\ref{eq:theorem1_goal}) for every $i = 1,\ldots,n$ by considering the following three cases: $\mathbb{E}\left[ 2Y_i-1\mid \mathcal{X}_n\right]=0$, $\mathbb{E}\left[ 2Y_i-1\mid\mathcal{X}_n\right]>0$, and $\mathbb{E}\left[ 2Y_i-1\mid\mathcal{X}_n\right]<0$. For every $i$ with $\mathbb{E}\left[ 2Y_i-1\mid\mathcal{X}_n\right]=0$, Eq. (\ref{eq:theorem1_goal}) holds with equality. For every $i$ with $\mathbb{E}\left[ 2Y_i-1\mid \mathcal{X}_n\right]>0$, we have $X_i\beta>0$ from Lemma \ref{Lemma:Conditional Moment Inequalities}, and therefore Eq. (\ref{theorem1_proof_eq_main}) implies $X_ib>0$, which in turn implies Eq. (\ref{eq:theorem1_goal}). For every $i$ with $\mathbb{E}\left[ 2Y_i-1\mid \mathcal{X}_n\right]<0$, we have $X_i\beta<0$ from Lemma \ref{Lemma:Conditional Moment Inequalities}, and therefore Eq. (\ref{theorem1_proof_eq_main}) implies $X_ib<0$, which in turn implies Eq. (\ref{eq:theorem1_goal}).
\end{proof}

\begin{proof}[Proof of Theorem \ref{Theorem: size control}]
If Eq. (\ref{eq:relat_barT_T}) holds under $H_0$, then $\mathbb{P}\left(T_n(\beta)\leq q_{1-\alpha}\mid \mathcal{X}_n\right)\geq\mathbb{P}\left(T^{\ast}_n(\beta)\leq q_{1-\alpha}\mid \mathcal{X}_n\right)\geq1-\alpha$. For the rest of the proof, we are going to show inequality (\ref{eq:relat_barT_T}) under $H_0$. Since 
$$
\begin{cases}
Y_i\geq Y^{\ast}_i & \mbox{ if }X_i\beta\geq 0 \\ 
Y_i\leq Y^{\ast}_i & \mbox{ if }X_i\beta\leq 0,
\end{cases}
$$
for every $i=1,\ldots,n$, we have 
$$
\mathbb{E}_n\left[ (2Y-1)1\{X\beta\geq 0, Xv<0\} \right]\geq \mathbb{E}_n\left[ (2 Y^{\ast}-1)1\{X\beta\geq 0 > Xv \} \right] \text{,  }\forall v\in \mathcal{V}_u
$$
and 
$$
\mathbb{E}_n\left[ (1-2Y)1\{X\beta\leq 0, Xv>0\} \right] \geq \mathbb{E}_n\left[(1-2 Y^{\ast})1\{X\beta\leq 0 < Xv \} \right] \text{,  }\forall v\in \mathcal{V}_l.
$$

By the construction of ${T}^{\ast}\left(\beta\right)$ and $T_{n}\left(\beta\right)$, it suffices to show that, for every $v\in \mathbb{R}^K$,  the two functions, 
$$
t\mapsto\max\left\{0,\sqrt{n}\frac{-t}{\max \{ \epsilon ,\sqrt{\mathbb{E}_n[1\{X\beta\geq 0, Xv< 0\} ]-t^2}\} }\right\}
$$
and 
$$
t\mapsto\max\left\{0,\sqrt{n}\frac{-t}{\max \{ \epsilon ,\sqrt{\mathbb{E}_n[1\{X\beta\leq 0, Xv>0\} ]-t^2}\} }\right\}
$$
are weakly decreasing. For the rest of the proof, we focus on the first function 
$$
f(t)\equiv\max\left\{0,\sqrt{n}\frac{-t}{\max \{ \epsilon ,\sqrt{\mathbb{E}_n[1\{X\beta\geq 0, Xv< 0\} ]-t^2}\} }\right\}.
$$
Consider $t_1$ and $t_2$ with $t_1<t_2$. If $t_2\geq 0$, we have $f(t_1)\geq 0=f(t_2)$. For the rest of the proof, therefore, we are going to show $f(t_1)\geq f(t_2)$ when $t_1<t_2<0$. Since $t_2^2<t_1^2$, we have 
$$
\mathbb{E}_n[1\{X\beta\geq 0, Xv< 0\} ]-t_1^2
<
\mathbb{E}_n[1\{X\beta\geq 0, Xv< 0\} ]-t_2^2,
$$
so 
$$
0
<
\max \left\{ \epsilon ,\sqrt{\mathbb{E}_n[1\{X\beta\geq 0, Xv< 0\} ]-t_1^2}\right\}
\leq
\max \left\{ \epsilon ,\sqrt{\mathbb{E}_n[1\{X\beta\geq 0, Xv< 0\} ]-t_2^2}\right\}.
$$
Since $-t_1>-t_2>0$, we have 
$$
\frac{-t_1}{\max \{ \epsilon ,\sqrt{\mathbb{E}_n[1\{X\beta\geq 0, Xv< 0\} ]-t_1^2}\}}
>
\frac{-t_2}{\max \{ \epsilon ,\sqrt{\mathbb{E}_n[1\{X\beta\geq 0, Xv< 0\} ]-t_2^2}\}}.
$$
Therefore, $f(t_1)>f(t_2)$. 
\end{proof}

\begin{proof}[Proof of Theorem \ref{theorem:power}]
By the choice of $cv$ and $q_{1-\alpha}$, we have $\mathbb{P}\left(T^{\ast}_n(\beta)\leq cv\mid \mathcal{X}_n\right)<1-\alpha-\eta$ for some positive number $\eta$. Consider a distribution of $(U_1,\ldots,U_n)$ given $\mathcal{X}_n$ under which $U_1,\ldots,U_n$ are independent given $\mathcal{X}_n$ and $U_i\mid \mathcal{X}_n\sim N(0,\sigma^2)$ for every $i=1,\ldots,n$, where $\sigma$ is chosen to satisfy 
$$
1-\prod_{i=1}^n\left(\Phi\left(\frac{\min\{-X_i\beta,0\}}{\sigma}\right)+1-\Phi\left(\frac{\max\{-X_i\beta,0\}}{\sigma}\right)\right)
\leq
\eta.
$$
(Such $\sigma$ exists because the left hand side converges to $0$ as $\sigma\rightarrow\infty$.) Note that if $Y_i=Y^{\ast}_i\mbox{ for every }i=1,\ldots,n$, then $T_n(\beta)=T^{\ast}_n(\beta)$. Then 
\begin{eqnarray*}
\mathbb{P}\left(T_n(\beta)\leq cv\mid \mathcal{X}_n\right)
&\leq&
\mathbb{P}\left(T_n(\beta)\leq cv,\ Y_i=Y^{\ast}_i\ \forall i=1,\ldots,n\mid \mathcal{X}_n\right)+1-\mathbb{P}\left(Y_i=Y^{\ast}_i\ \forall i=1,\ldots,n\mid \mathcal{X}_n\right)\\
&\leq&
\mathbb{P}\left(T^{\ast}_n(\beta)\leq cv\mid \mathcal{X}_n\right)+1-\mathbb{P}\left(Y_i=Y^{\ast}_i\ \forall i=1,\ldots,n\mid \mathcal{X}_n\right)\\
&<&
1-\alpha-\eta+1-\mathbb{P}\left(Y_i=Y^{\ast}_i\ \forall i=1,\ldots,n\mid \mathcal{X}_n\right).
\end{eqnarray*}
To show the statement of this theorem, the rest of the proof is going to show $\mathbb{P}\left(Y_i=Y^{\ast}_i\ \forall i=1,\ldots,n\mid \mathcal{X}_n\right)\geq 1-\eta$. Since $U_1,\ldots,U_n$ are independent given $\mathcal{X}_n$, the events $\{Y_1=Y^{\ast}_1\},\ldots,\{Y_n=Y^{\ast}_n\}$ are also independent given $\mathcal{X}_n$, and then $1-\mathbb{P}\left(Y_i=Y^{\ast}_i\ \forall i=1,\ldots,n\mid \mathcal{X}_n\right)=1-\prod_{i=1}^n\mathbb{P}\left(Y_i=Y^{\ast}_i\mid \mathcal{X}_n\right)$. Since the distribution of $U_i$ given $\mathcal{X}_n$ satisfies
\begin{eqnarray*}
\mathbb{P}\left(Y_i=Y^{\ast}_i\mid \mathcal{X}_n\right)
&=&
\mathbb{P}\left(1\{X_i\beta+U_i\geq 0\}=1\{U_i\geq 0\}\mid \mathcal{X}_n\right)\\
&=&
\mathbb{P}\left(U_i<\min\{-X_i\beta,0\}\mbox{ or }U_i\geq\max\{-X_i\beta,0\}\mid \mathcal{X}_n\right)\\
&=&
\Phi\left(\frac{\min\{-X_i\beta,0\}}{\sigma}\right)+1-\Phi\left(\frac{\max\{-X_i\beta,0\}}{\sigma}\right),
\end{eqnarray*}
we have 
$$
1-\mathbb{P}\left(Y_i=Y^{\ast}_i\ \forall i=1,\ldots,n\mid \mathcal{X}_n\right)
=
1-\prod_{i=1}^n\left(\Phi\left(\frac{\min\{-X_i\beta,0\}}{\sigma}\right)+1-\Phi\left(\frac{\max\{-X_i\beta,0\}}{\sigma}\right)\right)
\leq
\eta.
$$
\end{proof}

\begin{proof}[Proof of Theorem \ref{theorem:power_bound}]
In this proof, we focus on Eq. (\ref{eq:distant_alt}). Define $W= (2Y-1) 1\{Xb\geq 0,Xv<0\}$. First, we are going to show that
\begin{equation}\label{eq:imply_rela}
\sqrt{n}\mathbb{E}_n[ W]<-q_{1-\alpha}\max\left\{\epsilon,\sqrt{\frac{\mathbb{E}_n[W^2]}{1+q_{1-\alpha}^2/n}}\right\}
\implies 
T_n(b)>q_{1-\alpha}.
\end{equation}
Suppose $\sqrt{n}\mathbb{E}_n[ W]<-q_{1-\alpha}\max\left\{\epsilon,\sqrt{\frac{\mathbb{E}_n[W^2]}{1+q_{1-\alpha}^2/n}}\right\}$. Note that 
\begin{equation}\label{eq:negative_moment}
\mathbb{E}_n[ W]<0
\end{equation}
and 
$$
n\mathbb{E}_n[ W]^2>\max\left\{\epsilon^2 q_{1-\alpha}^2,\mathbb{E}_n[W^2]\frac{q_{1-\alpha}^2}{1+q_{1-\alpha}^2/n}\right\}.
$$
The second inequality implies $n\mathbb{E}_n[ W]^2>\max \{ \epsilon^2 q_{1-\alpha}^2,\mathbb{E}_n[W^2]q_{1-\alpha}^2-\mathbb{E}_n[ W]^{2}q_{1-\alpha}^2\}$. Using Eq. (\ref{eq:negative_moment}), $-\sqrt{n}\mathbb{E}_n[ W]>q_{1-\alpha}\max \{ \epsilon ,\sqrt{\mathbb{E}_n[W^2]-\mathbb{E}_n[ W] ^{2}}\}$ and then 
$$
\sqrt{n}\frac{-\mathbb{E}_n[ W]}{\max \{ \epsilon ,\sqrt{\mathbb{E}_n[W^2]-\mathbb{E}_n[ W]^{2}}\} }>q_{1-\alpha}
$$
which implies $T_n(b)>q_{1-\alpha}$. 

Then, we are going to show $\mathbb{P}(T_n(b)>q_{1-\alpha}\mid \mathcal{X}_n)\geq 1-\rho$. Using Eq. (\ref{eq:imply_rela}), we have 
$$
\mathbb{P}(T_n(b)>q_{1-\alpha}\mid \mathcal{X}_n)
\geq
\mathbb{P}\left(\sqrt{n}\mathbb{E}_n[ W]<-q_{1-\alpha}\max\left\{\epsilon,\sqrt{\frac{\mathbb{E}_n[W^2]}{1+q_{1-\alpha}^2/n}}\right\}\mid \mathcal{X}_n\right).
$$
Eq. (\ref{eq:distant_alt}) implies 
\begin{eqnarray*}
\mathbb{P}(T_n(b)>q_{1-\alpha}\mid \mathcal{X}_n)
&\geq&
\mathbb{P}\left(\mathbb{E}_n[W]<\mathbb{E}\left[\mathbb{E}_n[W]\mid \mathcal{X}_n\right]+\sqrt{\frac{2\log(1/\rho)\mathbb{E}_n[1\{Xb\geq 0,Xv<0\}]}{n}}\mid \mathcal{X}_n\right)\\
&=&
1-\mathbb{P}\left(\mathbb{E}_n[W]\geq\mathbb{E}\left[\mathbb{E}_n[W]\mid \mathcal{X}_n\right]+\sqrt{\frac{2\log(1/\rho)\mathbb{E}_n[1\{Xb\geq 0,Xv<0\}]}{n}}\mid \mathcal{X}_n\right).
\end{eqnarray*}
Since $-1\{Xb\geq 0,Xv<0\}\leq W_i\leq 1\{Xb\geq 0,Xv<0\}$ for every $i=1,\ldots,n$, \cite{Hoeffding:1963}'s inequality implies 
$$
\mathbb{P}(T_n(b)>q_{1-\alpha}\mid \mathcal{X}_n)
\geq
1-\exp\left(-\frac{2n^2\left(\sqrt{\frac{2\log(1/\rho)\mathbb{E}_n[1\{Xb\geq 0,Xv<0\}]}{n}}\right)^2}{4\sum_{i=1}^n1\{X_ib\geq 0,X_iv<0\}}\right)
=
1-\rho.
$$
\end{proof}
\begin{proof}[Proof of Corollary \ref{Corollary: Power Guarantee 1}]
Let $1 - \rho$ denote the maximum of the power bounds \eqref{power bound upper} and \eqref{power bound lower} in the statement of the corollary.  If the expressions $\max\{0,\cdot\}$ in \eqref{power bound upper} and \eqref{power bound lower} are both zero for all $v \in \mathcal{V}_u$ and $v \in \mathcal{V}_l$, then the implication of the corollary is trivially satisfied. Thus suppose instead that the maximum of \eqref{power bound upper} and \eqref{power bound lower} is greater than zero. It follows that there is either a $v \in \mathcal{V}_u$ such that
\begin{equation*}
1-\rho = 1-\exp\left(-\frac{1}{2}\left(\sqrt{n}\tilde{\zeta}_u(b,v)-q_{1-\alpha}\max\left\{\tilde{\epsilon}_u(b,v),(1+q_{1-\alpha}^2/n)^{-1/2}\right\}\right)^2
\right)\text{,}
\end{equation*}
or a $v \in \mathcal{V}_l$ such that
\begin{equation*}
1-\rho = 1-\exp\left(-\frac{1}{2}\left(\sqrt{n}\tilde{\zeta}_l(b,v)-q_{1-\alpha}\max\left\{\tilde{\epsilon}_l(b,v),(1+q_{1-\alpha}^2/n)^{-1/2}\right\}\right)^2
\right)\text{,}
\end{equation*}
implying \eqref{power bound upper} in the former case and \eqref{power bound lower} in the latter. The conclusion of the corollary then follows from Theorem \ref{theorem:power_bound}.
\end{proof}

\begin{proof}[Proof of Corollary \ref{power corollary}]
The first part of the corollary can be shown by first noting that $Q(b) \geq C(\gamma)$ implies that at least one of the following inequalities hold:
\begin{align*}
\mathbb{E}\left[\ \mathbb{E}_n\left[ (2Y-1) 1\{Xb\geq 0, Xv<0\}\right]\mid \mathcal{X}_n\right] &\leq -C(\gamma)\text{,}\\ \mathbb{E}\left[\ \mathbb{E}_n\left[ (1-2Y) 1\{Xb\leq 0, Xv\geq0\}\right]\mid \mathcal{X}_n\right] &\leq -C(\gamma)\text{.}
\end{align*}
Then because $-C(1 - \rho)$ is less than or equal to each of the expressions on the right hand side of inequalities \eqref{eq:distant_alt} and \eqref{eq:distant_alt2}, at least one of \eqref{eq:distant_alt} and \eqref{eq:distant_alt2} is true and Theorem \ref{theorem:power_bound} delivers the result.

\noindent For the second claim of the corollary, consider first that if $\sqrt{n}Q(b) \leq q_{1-\alpha}\max\left\{\epsilon, (1+q_{1-\alpha}^2/n)^{-1/2}\right\}$ the result holds trivially. Thus, suppose instead that $\sqrt{n}Q(b) > q_{1-\alpha}\max\left\{\epsilon, (1+q_{1-\alpha}^2/n)^{-1/2}\right\}$ and consider
$$
\gamma = 1-\exp \left(-\frac{1}{2}\left(\sqrt{n}Q(b)-q_{1-\alpha}\max\left\{\epsilon, (1+q_{1-\alpha}^2/n)^{-1/2}\right\}\right)^2 \right).
$$
Then $Q(b) = C(\gamma)$ and the desired implication follows from the first part of the corollary.
\end{proof}

\begin{proof}[Proof of Theorem \ref{Theorem: likelihood ratio test most powerful with composite null}]
Assume w.l.o.g. that $\{i: \bar{p}_i\ne\tilde{p}_i\}=\{1,\ldots,\bar{n}\}$, which can be achieved by rearranging $i$'s.  Then $(\bar{p}_1,\ldots,\bar{p}_{\bar{n}})=(1/2,\ldots,1/2)$ and the test in \eqref{likelihood ratio test with composite null} simplifies as
$$
\bar{\phi}(y_1,\ldots,y_{\bar{n}})
=
\begin{cases}
1&\mbox{ if }\prod_{i=1}^{\bar{n}}\tilde{p}_i^{y_i}(1-\tilde{p}_i)^{1-y_i}>k2^{-\bar{n}}\text{,}\\
\xi&\mbox{ if }\prod_{i=1}^{\bar{n}}\tilde{p}_i^{y_i}(1-\tilde{p}_i)^{1-y_i}=k2^{-\bar{n}}\text{,}\\
0&\mbox{ otherwise}.
\end{cases}
$$
Define
$$
\mathrm{RP}(p_1,\ldots,p_{\bar{n}})\equiv\sum_{(y_1,\ldots,y_{\bar{n}})\in \{0,1\}^{\bar{n}}}\bar{\phi}(y_1,\ldots,y_{\bar{n}})\prod_{i=1}^{\bar{n}}{p}_i^{y_i}(1-{p}_i)^{1-y_i}.
$$
For the rest of the proof, we first show that $\mathrm{RP}(p_1,\ldots,p_{\bar{n}})$ is the probability that $\bar{\phi}(y_1,\ldots,y_{n})$ rejects the null hypothesis when $\mathbb{P}\{Y_i = 1\mid \mathcal{X}_n\} = p_i$ for all $i$, that is, 
\begin{equation}\label{equivalence of RP forumulations}
\mathrm{RP}(p_1,\ldots,p_{\bar{n}})=\sum_{(y_1,\ldots,y_n) \in \{0,1\}^n}\bar{\phi}(y_1,\ldots,y_{n})\prod_{i=1}^{n}{p}_i^{y_i}(1-{p}_i)^{1-y_i},
\end{equation}
where the right-hand side of the above equation is the sum of $2^n$ terms instead of $2^{\bar{n}}$. From this, it follows from \eqref{size control composite alternative} that $\mathrm{RP}(\bar{p}_1,\ldots,\bar{p}_{\bar{n}}) = \alpha$. Subsequently we show that $(\bar{p}_1,\ldots,\bar{p}_{\bar{n}})$ is the constrained maximizer of $\mathrm{RP}(p_1,\ldots,p_{\bar{n}})$ with respect to $p_1,\ldots,p_{\bar{n}}$, subject to $p_1,\ldots,p_{\bar{n}}$ being compatible with the null hypothesis. Thus \eqref{likelihood ratio test with composite null} achieves size control, i.e., $\sum_{(y_1,\ldots,y_n) \in \{0,1\}^n}\bar{\phi}(y_1,\ldots,y_{n})\prod_{i=1}^{n}{p}_i^{y_i}(1-{p}_i)^{1-y_i}\leq \alpha$ for any sequence $(p_1,\ldots,p_n)$ under the null hypothesis. By Theorem 3.8.1 of \cite{Lehmann/Romano:05} the test \eqref{likelihood ratio test with composite null} is a most powerful test of the composite null $\beta = b$ against the simple alternative $H_1$, and \eqref{LFD under composite null} is the least favorable distribution of $\left(Y_1,\ldots,Y_n\right)$ given $\mathcal{X}_n$.

First,  to establish \eqref{equivalence of RP forumulations}, we have
\begin{eqnarray*}
\mathrm{RP}(p_1,\ldots,p_{\bar{n}})
&=&
\sum_{(y_1,\ldots,y_{\bar{n}})\in \{0,1\}^{\bar{n}}}\bar{\phi}(y_1,\ldots,y_{\bar{n}})\prod_{i=1}^{\bar{n}}{p}_i^{y_i}(1-{p}_i)^{1-y_i}\\
&=&
\sum_{(y_1,\ldots,y_{\bar{n}})\in \{0,1\}^{\bar{n}}}\bar{\phi}(y_1,\ldots,y_{\bar{n}})\sum_{(y_{\bar{n}+1},\ldots,y_n)}\prod_{i=1}^n{p}_i^{y_i}(1-{p}_i)^{1-y_i}\\
&=&
\sum_{(y_1,\ldots,y_n) \in \{0,1\}^n}\bar{\phi}(y_1,\ldots,y_{n})\prod_{i=1}^{n}{p}_i^{y_i}(1-{p}_i)^{1-y_i},
\end{eqnarray*}
where the second equality holds because $\sum_{(y_{\bar{n}+1},\ldots,y_n)}\prod_{i=\bar{n}+1}^n{p}_i^{y_i}(1-{p}_i)^{1-y_i}=1$ and the third equality follows because $\bar{\phi}(y_1,\ldots,y_n)$ does not depend on $y_{\bar{n}+1},\ldots,y_n$.

For the next step of the proof, notice that \eqref{p_i_bar_definition} implies that, under the null hypothesis, $(\mathbb{P}\left(Y_i=1\mid \mathcal{X}_n\right)-1/2)(\tilde{p}_i-1/2)\leq 0$ for every $i=1,\ldots,\bar{n}$.  Thus we next show that $(\bar{p}_1,\ldots,\bar{p}_{\bar{n}})$ is the constrained maximizer of $\mathrm{RP}(p_1,\ldots,p_{\bar{n}})$ subject to $(p_i-1/2)(\tilde{p}_i-1/2)\leq 0$ for all $i=1,\ldots,\bar{n}$. Let $\bar{\phi}(d,y_{-j})$ be the shorthand for $\bar{\phi}(y_1,\ldots,y_{j-1},d,y_{j+1},\ldots,y_n)$, where $d\in\{0,1\}$. Note that
$$
\begin{cases}
\bar{\phi}(1,y_{-j})-\bar{\phi}(0,y_{-j})\geq 0&\mbox{ if }\tilde{p}_j>1/2\text{,}\\
\bar{\phi}(1,y_{-j})-\bar{\phi}(0,y_{-j})\leq 0&\mbox{ if }\tilde{p}_j<1/2\text{,}
\end{cases}
$$
for every $j=1,\ldots,\bar{n}$, and  that $\tilde{p}_j \neq 1/2$ because $\bar{p}_j \neq \tilde{p}_j$ for all $j=1,\ldots,\bar{n}$.
Since 
\begin{eqnarray*}
\frac{\partial}{\partial p_j}\mathrm{RP}(p_1,\ldots,p_{\bar{n}})
&=&
\sum_{(y_1,\ldots,y_{\bar{n}})\in \{0,1\}^{\bar{n}}}(2y_j-1)\bar{\phi}(y_1,\ldots,y_{\bar{n}})\prod_{i\ne j}{p}_i^{y_i}(1-{p}_i)^{1-y_i}
\\
&=&
\sum_{y_{-j}}\bar{\phi}(1,y_{-j})\prod_{i\ne j}{p}_i^{y_i}(1-{p}_i)^{1-y_i}-\sum_{y_{-j}}\bar{\phi}(0,y_{-j})\prod_{i\ne j}{p}_i^{y_i}(1-{p}_i)^{1-y_i}
\\
&=&
\sum_{y_{-j}}(\bar{\phi}(1,y_{-j})-\bar{\phi}(0,y_{-j}))\prod_{i\ne j}{p}_i^{y_i}(1-{p}_i)^{1-y_i},
\end{eqnarray*}
we have 
$$
\begin{cases}
\frac{\partial}{\partial p_j}\mathrm{RP}(p_1,\ldots,p_{\bar{n}})\geq 0&\mbox{ if }\tilde{p}_j>1/2\text{,}\\
\frac{\partial}{\partial p_j}\mathrm{RP}(p_1,\ldots,p_{\bar{n}})\leq 0&\mbox{ if }\tilde{p}_j<1/2.
\end{cases}
$$
Thus $(1/2,\ldots,1/2)$ maximizes $\mathrm{RP}(p_1,\ldots,p_{\bar{n}})$ with respect to $(p_1,\ldots,p_{\bar{n}})$ subject to $(p_i-1/2)(\tilde{p}_i-1/2)\leq 0$ for all $i=1,\ldots,\bar{n}$, completing the proof.
\end{proof}
\begin{proof}[Proof of Corollary \ref{corollary:trivial_LRT}]
The corollary follows directly from Theorem \ref{Theorem: likelihood ratio test most powerful with composite null} and the reasoning given in the text.
\end{proof}

\newpage
\section{Additional Simulation Results}\label{sec:additional_results}
\begin{figure}[ht]
\parbox{.4\textwidth}{
\centering
\includegraphics[width=.45\textwidth,keepaspectratio]{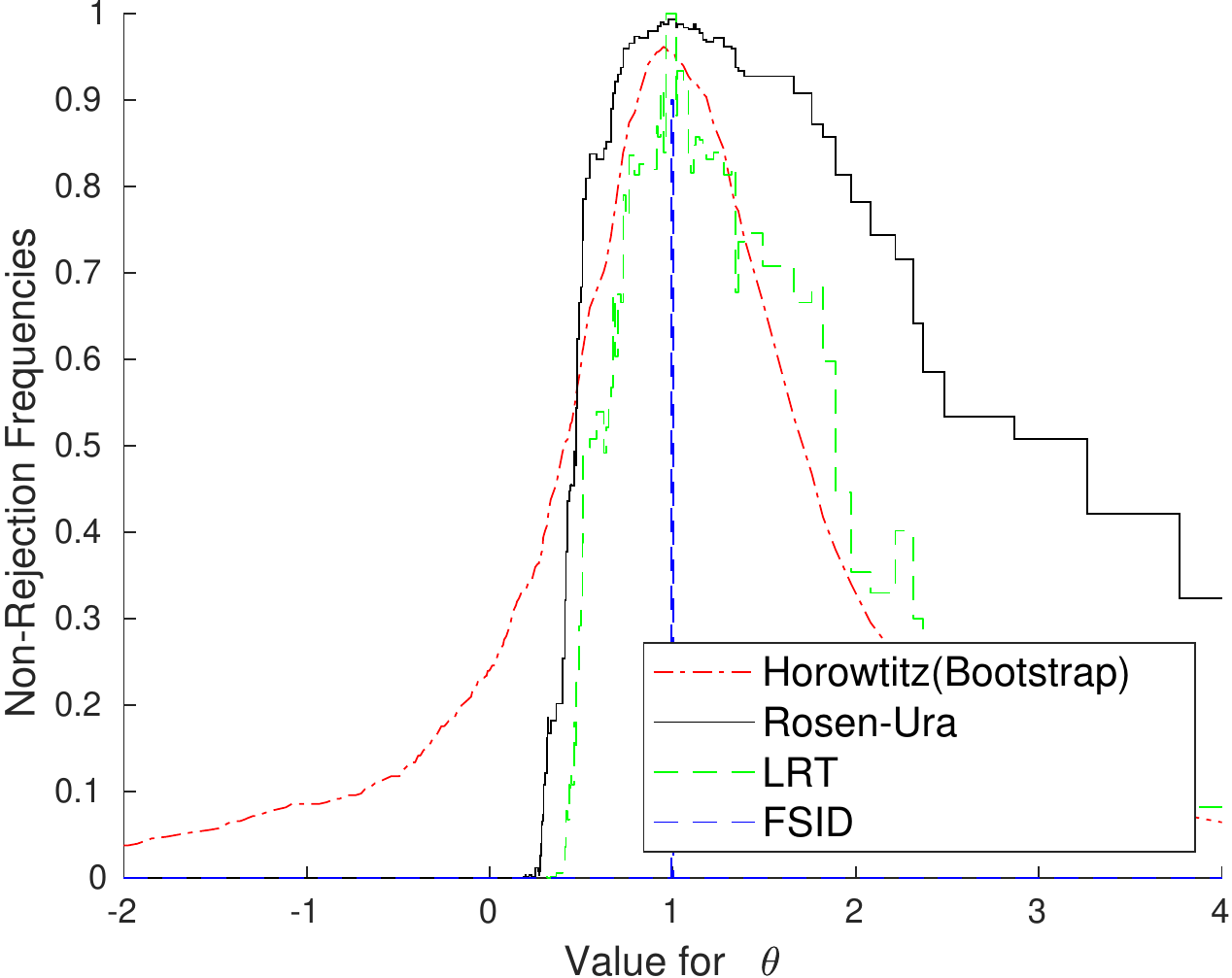}
\caption*{Figure \ref{fig_MC_app1}.a:  Design 1.}
}
\hspace{.05\textwidth}
\parbox{.4\textwidth}{
\centering
\includegraphics[width=.45\textwidth,keepaspectratio]{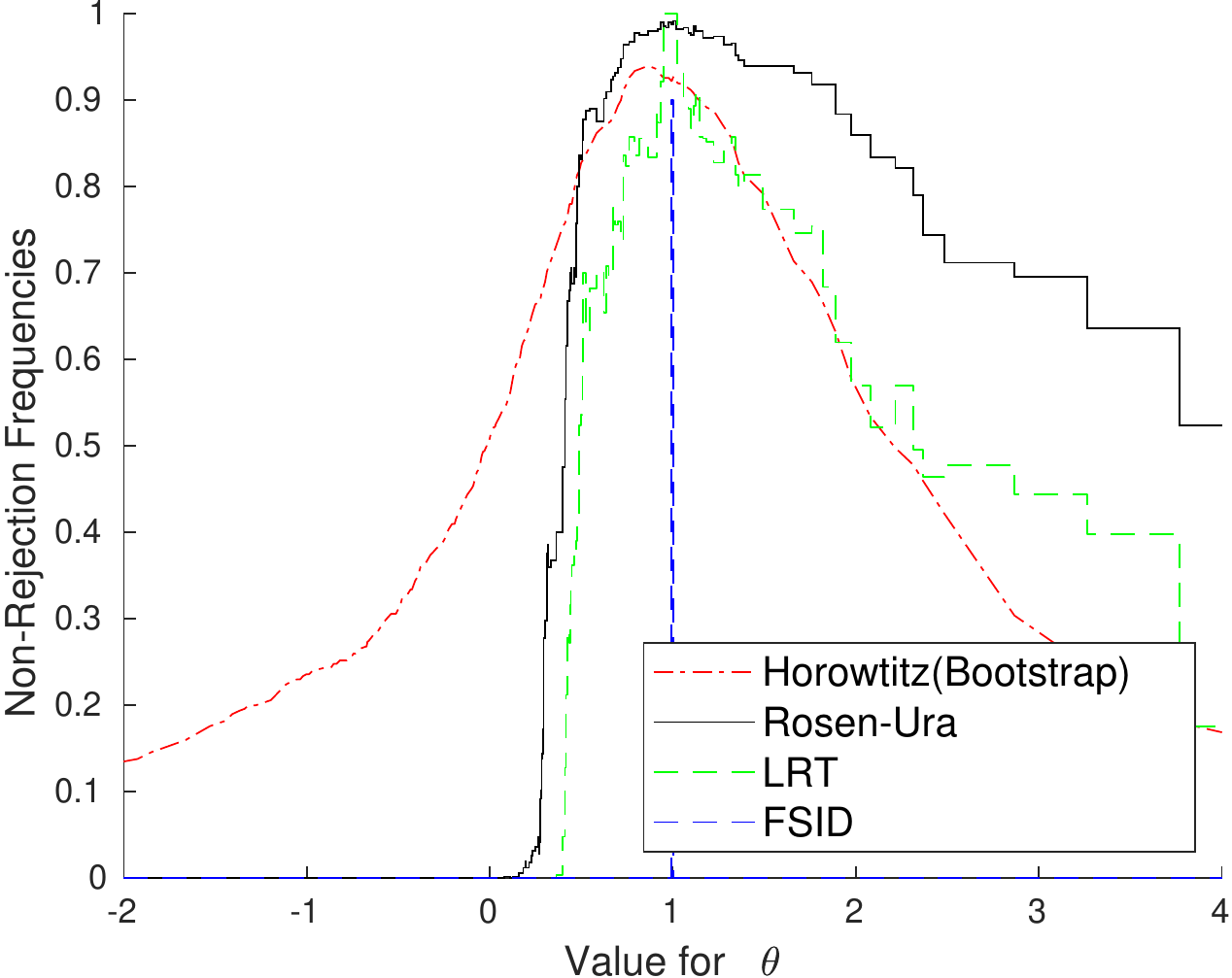}
\caption*{Figure \ref{fig_MC_app1}.b:  Design 2.}
}
\\
\bigskip
\\
\parbox{.4\textwidth}{
\centering
\includegraphics[width=.45\textwidth,keepaspectratio]{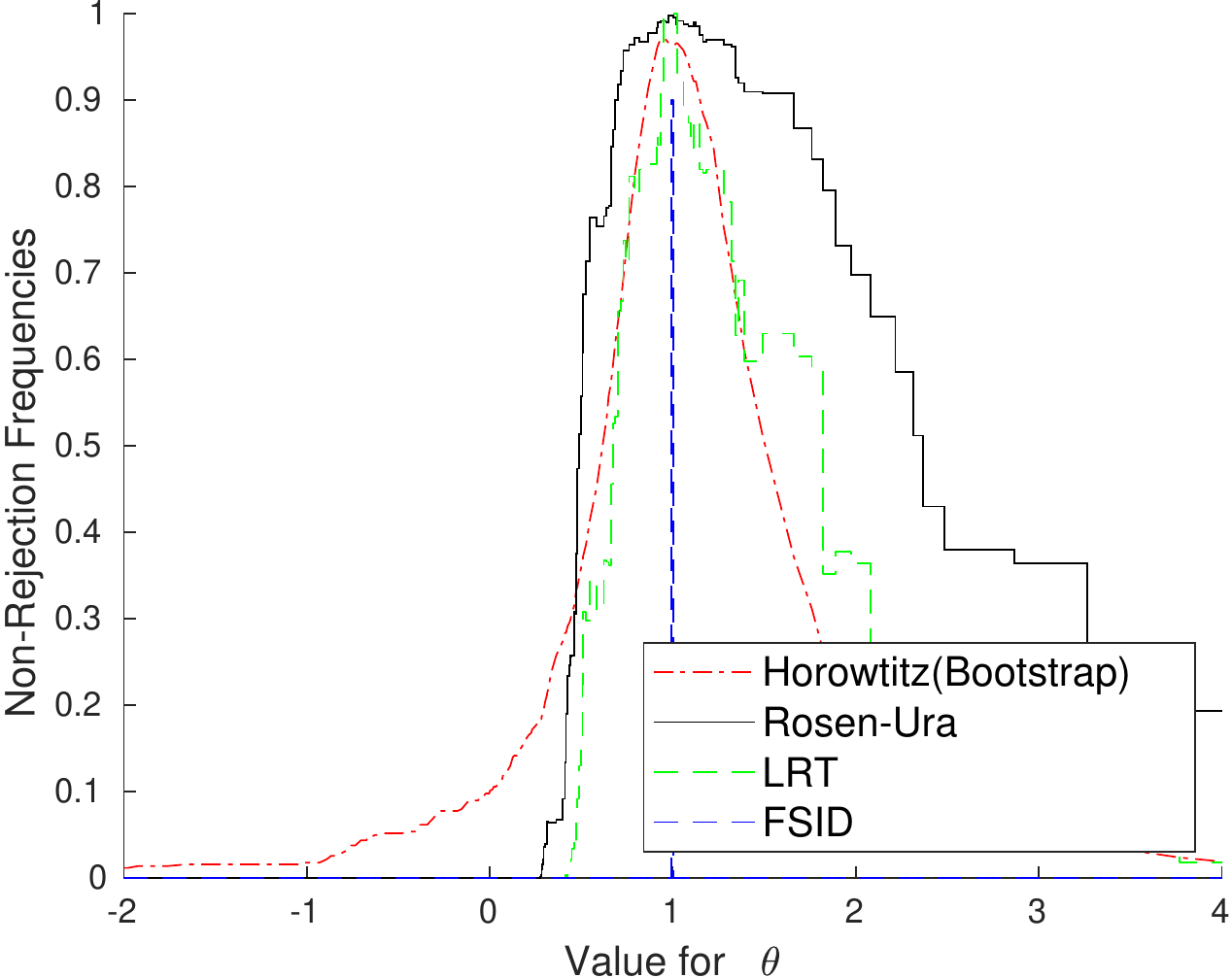}
\caption*{Figure \ref{fig_MC_app1}.c:  Design 3.}
}
\hspace{.05\textwidth}
\parbox{.4\textwidth}{
\centering
\includegraphics[width=.45\textwidth,keepaspectratio]{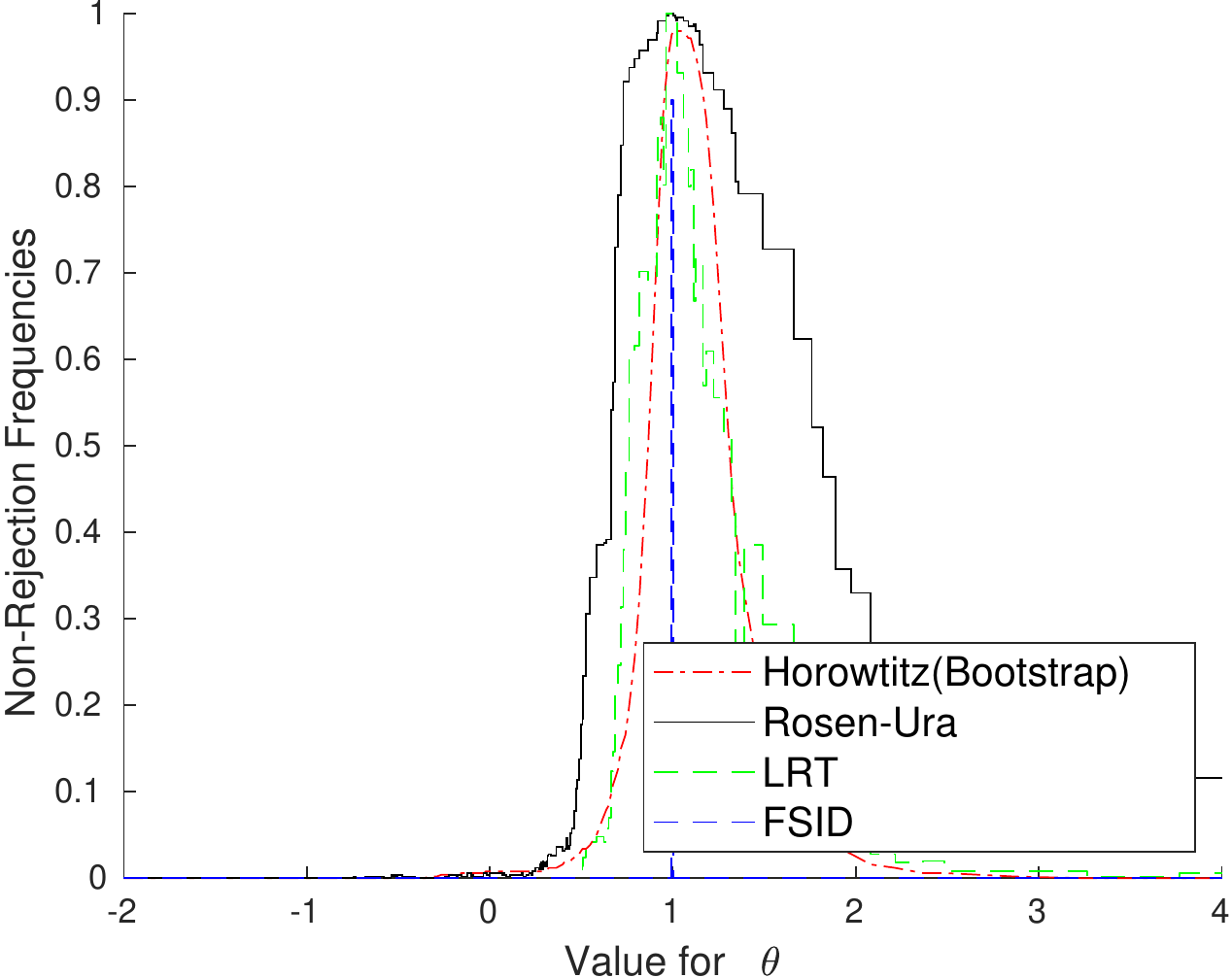}
\caption*{Figure \ref{fig_MC_app1}.d:  Design 4.}
}
\\
\\
\caption{Non-rejection frequencies with $1-\protect\alpha=90\%$ with true $\theta_0=1$ and $n=250$. }
\label{fig_MC_app1}
\end{figure}

\begin{figure}[ht]
\parbox{.4\textwidth}{
\centering
\includegraphics[width=.45\textwidth,keepaspectratio]{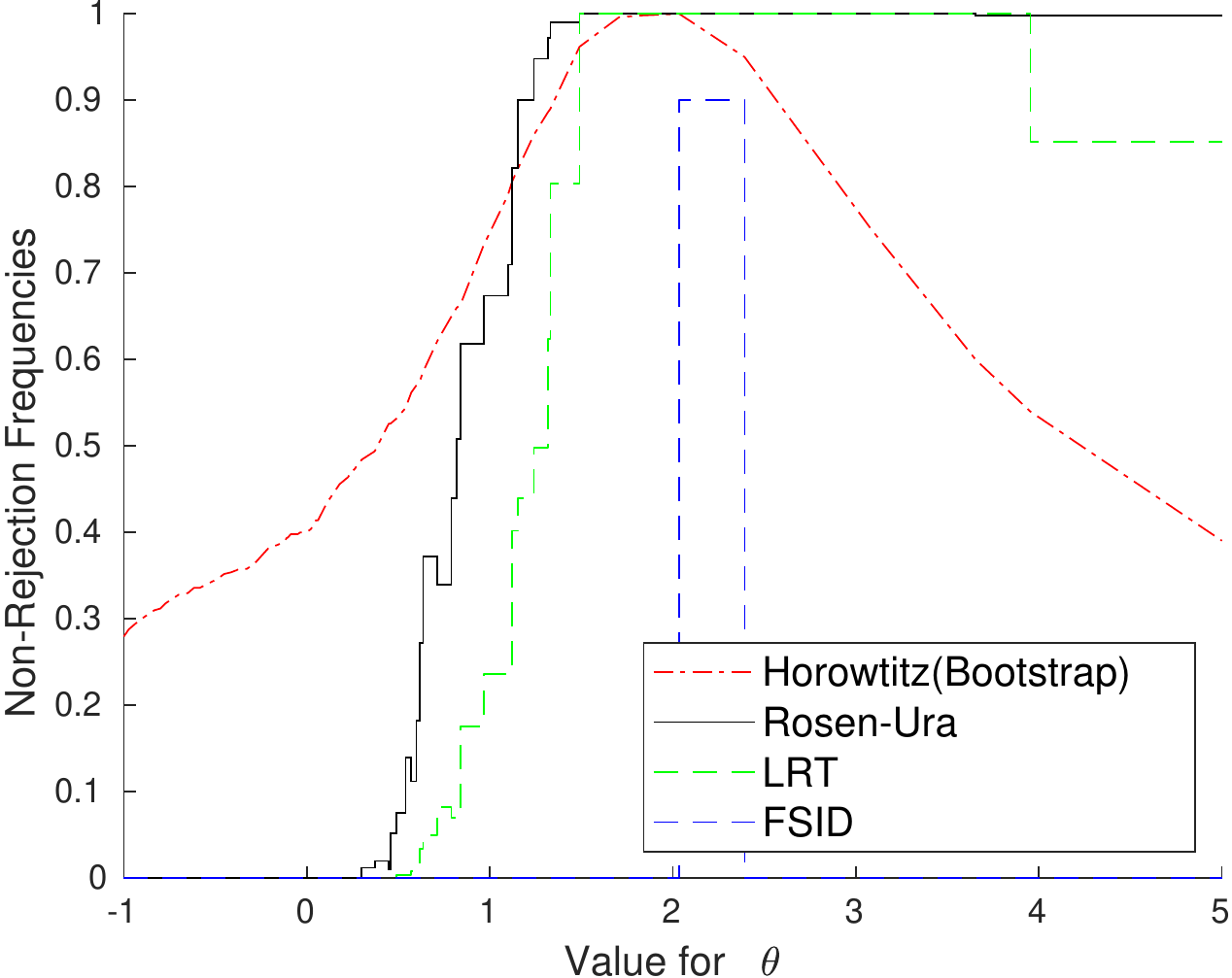}
\caption*{Figure \ref{fig_MC_app2}.a:  Design 1.}
}
\hspace{.05\textwidth}
\parbox{.4\textwidth}{
\centering
\includegraphics[width=.45\textwidth,keepaspectratio]{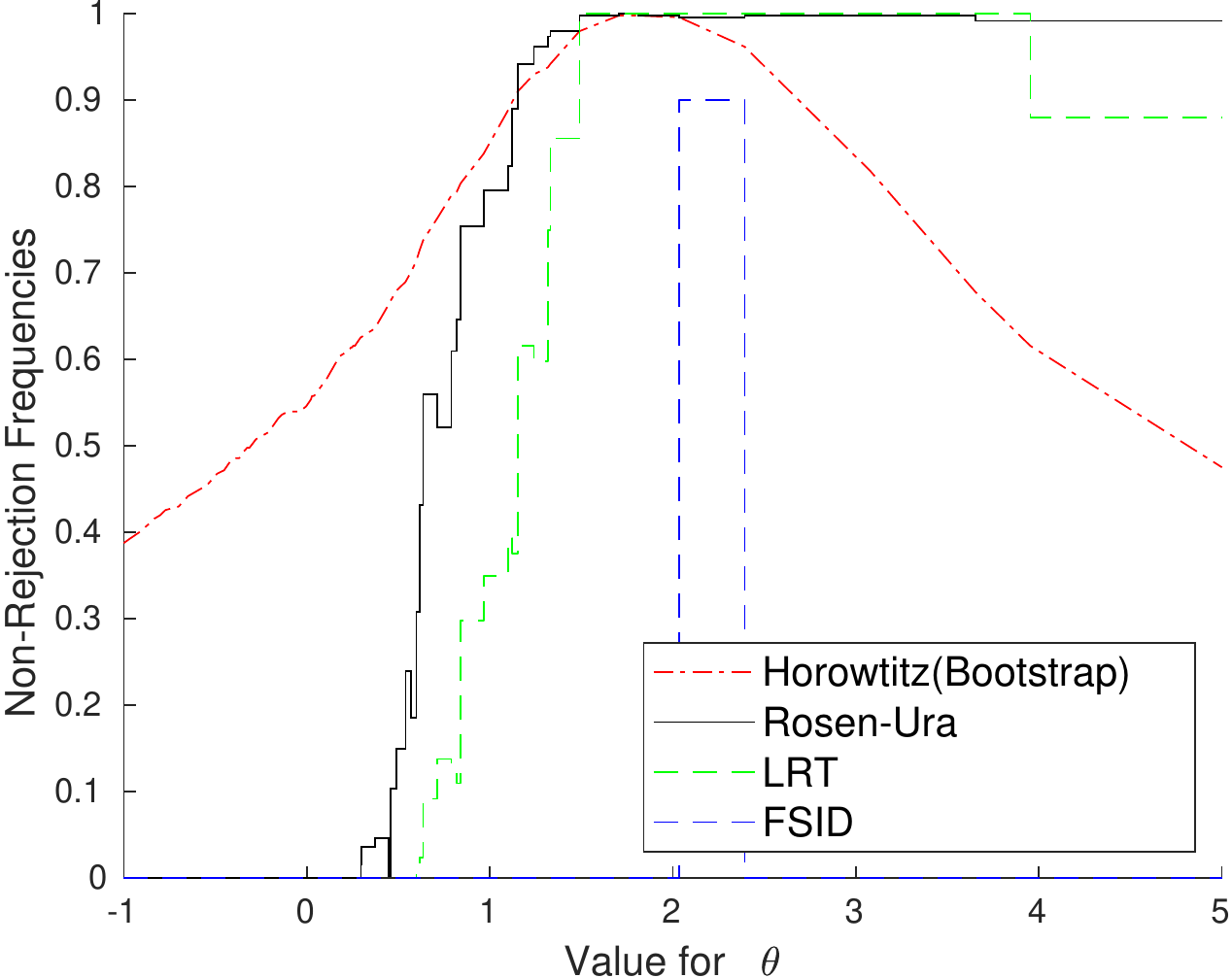}
\caption*{Figure \ref{fig_MC_app2}.b:  Design 2.}
}
\\
\bigskip
\\
\parbox{.4\textwidth}{
\centering
\includegraphics[width=.45\textwidth,keepaspectratio]{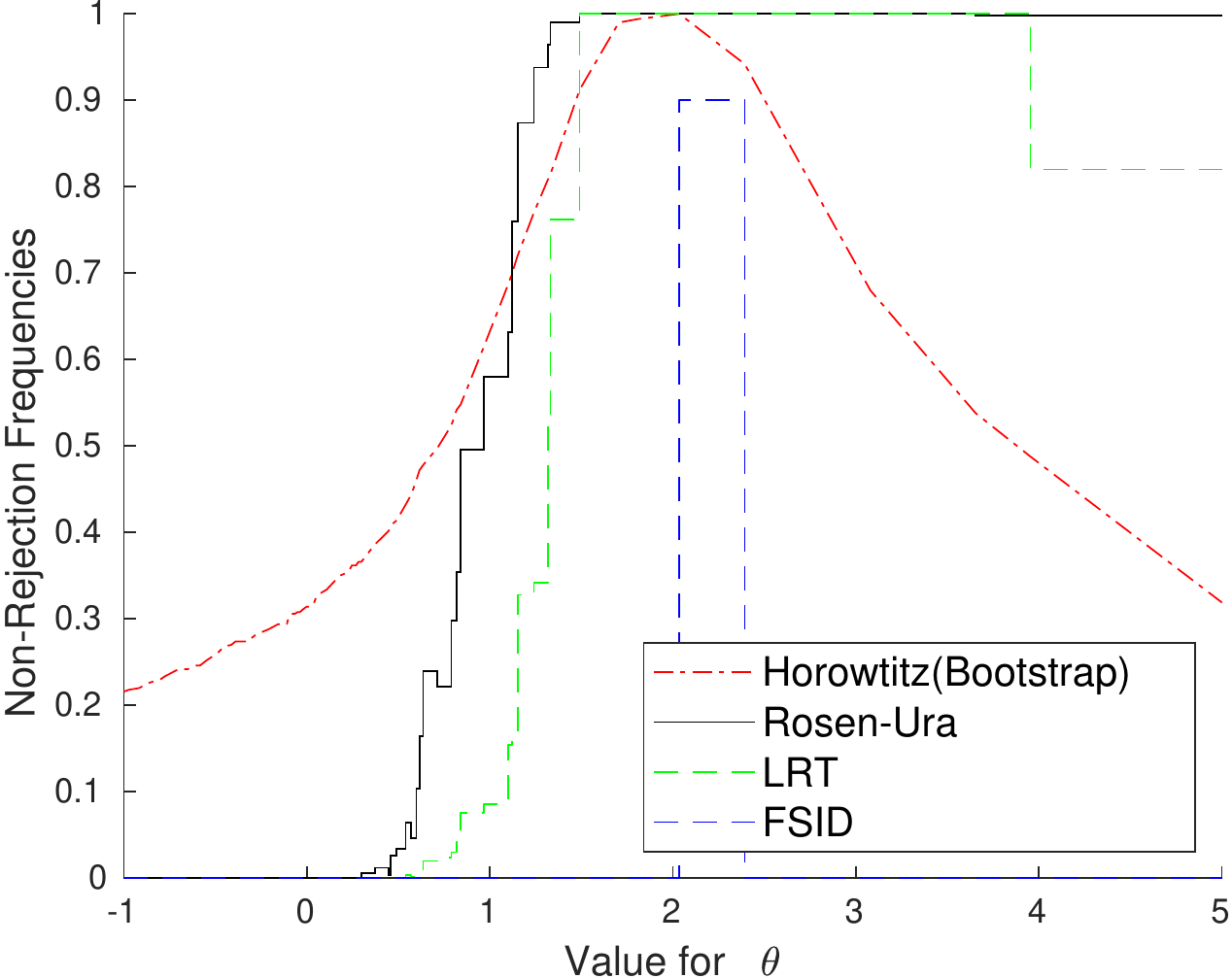}
\caption*{Figure \ref{fig_MC_app2}.c:  Design 3.}
}
\hspace{.05\textwidth}
\parbox{.4\textwidth}{
\centering
\includegraphics[width=.45\textwidth,keepaspectratio]{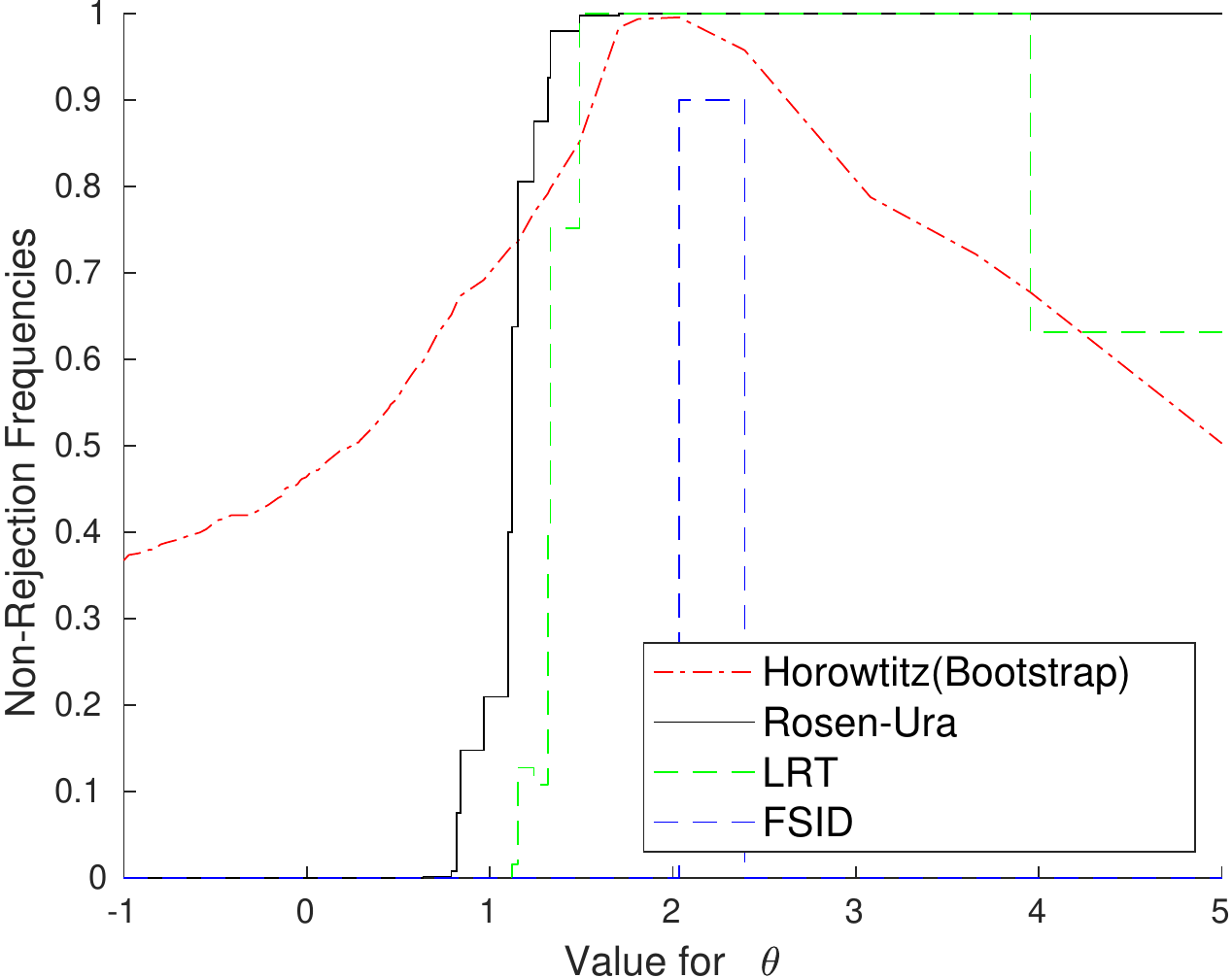}
\caption*{Figure \ref{fig_MC_app2}.d:  Design 4.}
}
\\
\\
\caption{Non-rejection frequencies with $1-\protect\alpha=90\%$ with true $\theta_0=2$ and $n=100$. }
\label{fig_MC_app2}
\end{figure}

\begin{figure}[ht]
\parbox{.4\textwidth}{
\centering
\includegraphics[width=.45\textwidth,keepaspectratio]{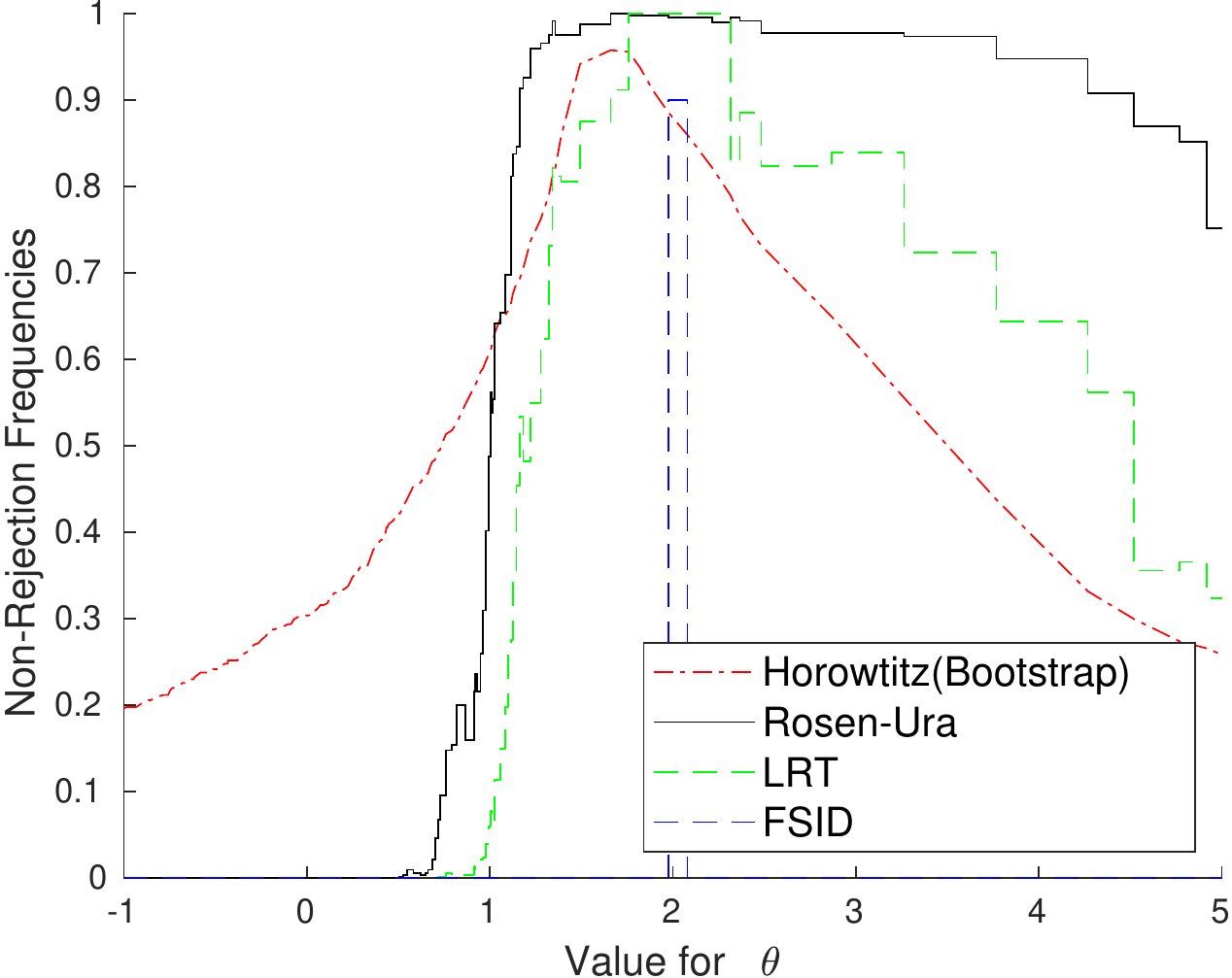}
\caption*{Figure \ref{fig_MC_app3}.a:  Design 1.}
}
\hspace{.05\textwidth}
\parbox{.4\textwidth}{
\centering
\includegraphics[width=.45\textwidth,keepaspectratio]{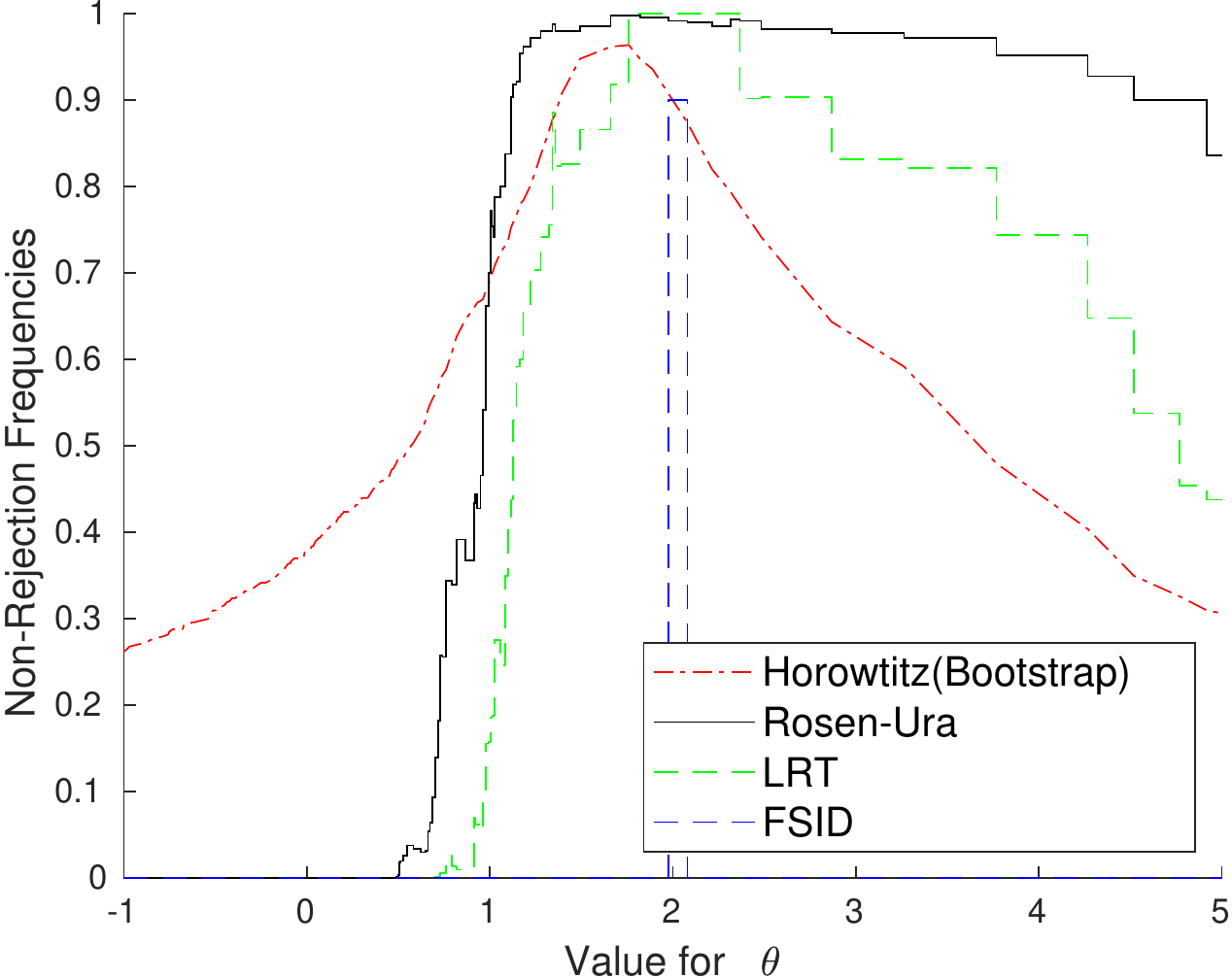}
\caption*{Figure \ref{fig_MC_app3}.b:  Design 2.}
}
\\
\bigskip
\\
\parbox{.4\textwidth}{
\centering
\includegraphics[width=.45\textwidth,keepaspectratio]{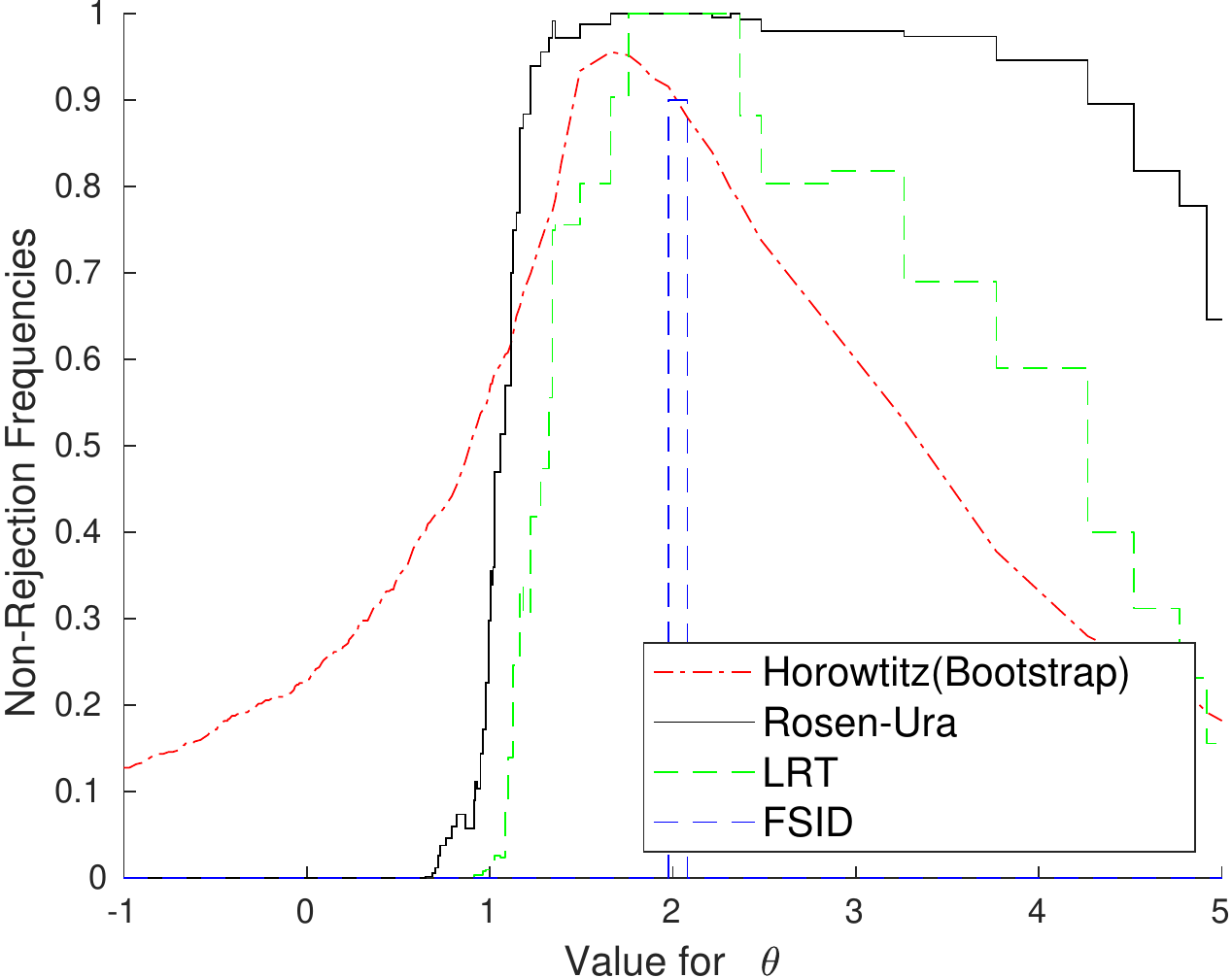}
\caption*{Figure \ref{fig_MC_app3}.c:  Design 3.}
}
\hspace{.05\textwidth}
\parbox{.4\textwidth}{
\centering
\includegraphics[width=.45\textwidth,keepaspectratio]{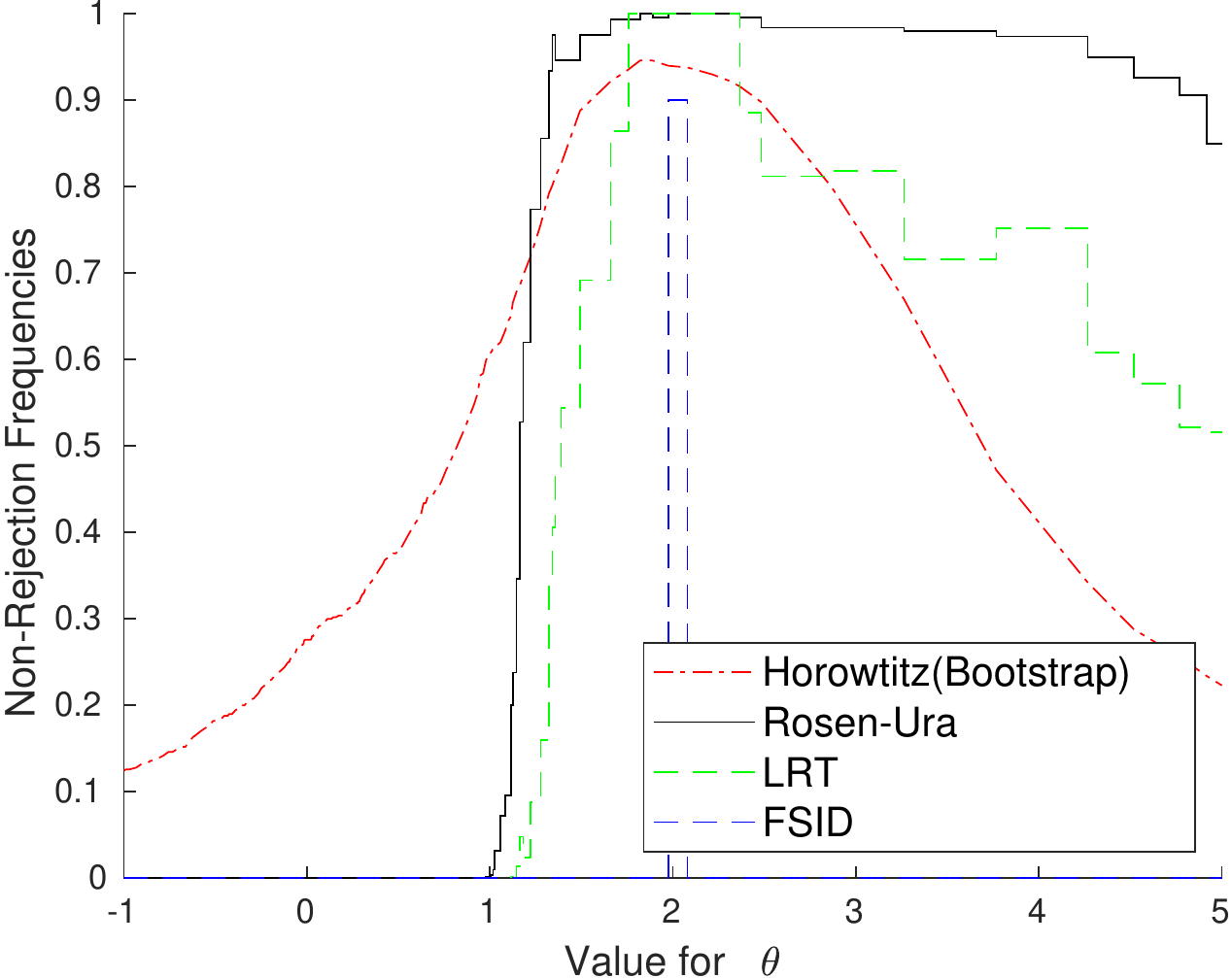}
\caption*{Figure \ref{fig_MC_app3}.d:  Design 4.}
}
\\
\\
\caption{Non-rejection frequencies with $1-\protect\alpha=90\%$ with true $\theta_0=2$ and $n=250$. }
\label{fig_MC_app3}
\end{figure}

\end{document}